\documentclass[a4paper,12pt]{article}

\usepackage{amsmath,amssymb,mathrsfs}
\usepackage[dvips]{graphicx}
\usepackage{amsthm,bm,
color}

\makeatletter
\def\figcaption{\def\@captype{figure}\caption}
\makeatother

\makeatletter
 
  \@addtoreset{equation}{section}
 \makeatother

\newtheorem{theorem}{\bf Theorem}[section]
\newtheorem{lemma}[theorem]{\bf Lemma}
\newtheorem{proposition}[theorem]{\bf Proposition}
\newtheorem{corollary}[theorem]{\bf Corollary}

\newcommand{\ro}[1]{\expandafter{\romannumeral#1}}
\newcommand{\Ro}[1]{\uppercase\expandafter{\romannumeral#1}}

\title{Spectrum of the Laplacian\\ on a covering graph with pendant edges I:\\
	The one-dimensional lattice and beyond}
\author{
Akito Suzuki\thanks{  e-mail: akito@shinshu-u.ac.jp}\\
        Department of Mathematics,\\
        Faculty of Engineering,
        Shinshu University,\\
        Wakasato, Nagano 380, Japan
}
\begin{document}
\maketitle
\vspace{-18mm}
\tableofcontents
\vspace{10mm}
\begin{abstract}
In this paper, we examine covering graphs
that are obtained from the $d$-dimensional integer lattice by adding pendant edges.
In the case of $d=1$, 
we show that the Laplacian on the graph has a spectral gap
and establish a necessary and sufficient condition under which the Laplacian has no eigenvalues.
In the case of $d=2$,
we show that there exists an arrangement of the pendant edges such that 
the Laplacian has no spectral gap.
\end{abstract}
\section{Introduction} 

\subsection{Definitions and examples}

Let $G=(V(G), E(G))$ be an undirected graph,
where $V(G)$ and $E(G)$ are the sets of vertices and edges, respectively. 
We denote an edge that connects vertices $x$ and $y$ as an unordered pair $\{x,y\}$ 
and write $x \sim y$ if $\{x,y\} \in E(G)$.
For a vertex $x \in V(G)$, 
we denote 
the number of edges incident to $x$ as ${\rm deg}x$, i.e., ${\rm deg} x = \# \{ y \mid y \sim x \}$.
A vertex $x$ is called an {\it end vertex} if ${\rm deg}x = 1$, 
and an edge that contains an end vertex is called a {\it pendant edge}. 
In this paper, we first consider a graph satisfying the following conditions:
\begin{itemize}
\item[(G$_1$)] $\{ (n,0) \in \mathbb{Z} \times \{ 0,1 \} \mid n \in \mathbb{Z} \} \subset V(G) 
	\subset \mathbb{Z} \times \{ 0,1 \}$ and a vertex $(n,0) \in V(G)$ is connected to $(n \pm 1,0) \in V(G)$.
\item[(G$_2$)] 
	If $(n,1) \in V(G)$, then $(n,1)$ is an end vertex connected to $(n,0)$ and ${\rm deg}(n,0)=3$.
	If $(n,1) \not\in V(G)$, then ${\rm deg}(n,0)=2$.
\end{itemize}
Such a graph is obtained by adding pendant edges 
to the one-dimensional lattice $\mathbb{Z}$,
in which case a vertex $n \in \mathbb{Z}$ and an end vertex connected to $n \in \mathbb{Z}$ 
are identified with $(n,0) \in V(G)$ 
and $(n,1) \in V(G)$, respectively.
Let
\[ \ell^2(V(G)) = \left\{ \psi:V(G) \to \mathbb{C} ~\Big|~ \langle \psi | \psi \rangle < \infty \right\} \]
be the Hilbert space of square summable functions on $V(G)$ with the inner product:
\[ \langle \psi | \phi \rangle = \sum_{x \in V(G)} \overline{\psi}(x) \phi(x) {\rm deg} x,
	\quad \psi, \phi \in \ell^2(V(G)). \]
The transition operator $L_G$ is a bounded self-adjoint operator on $\ell^2(V(G))$
defined as
\[ (L_G \psi)(x) = \frac{1}{{\rm deg} x} \sum_{y \sim x} \psi(y), \quad \psi \in \ell^2(V(G)). \]
We are interested in the spectrum $\sigma(-\Delta_G)$ of the (negative) Laplacian $- \Delta_G$ 
defined by $- \Delta_G = 1- L_G$.
From the spectral mapping theorem, we know that $\sigma(-\Delta_G) = \{ 1-l \mid l \in \sigma(L_G) \}$
and that $\sigma_{\rm  p}(-\Delta_G)= \{ 1-l \mid l \in \sigma_{\rm p}(L_G) \}$,
where $\sigma_{\rm p}(L)$ denotes the set of eigenvalues of the operator $L$.
Hence, we know that $L_G$ has a spectral gap (resp. an eigenvalue) if and only if
$-\Delta_G$ has a spectral gap (resp. an eigenvalue).
In this paper, we consider the spectrum $\sigma(L_G)$ of the transition operator $L_G$.

A trivial example of a graph $G$ with properties (G$_1$) and (G$_2$) is the one dimensional lattice $\mathbb{Z}$,
{\it i.e.} $G$ has no pendant edges.
It is well known and easy to show that $\sigma(L_\mathbb{Z}) = [-1,1]$
and that $L_\mathbb{Z}$ has no eigenvalues.

Adding a single pendant edge to each vertex of the lattice $\mathbb{Z}$,
we obtain a graph $G_{1,1}$ such that $V(G_{1,1})= \mathbb{Z} \times \{0,1\}$
(see Figure \ref{fig00},
where we denote a vertex identified with an integer and 
an end vertex by $\bigcirc$ and $\triangle$, respectively).
$G_{1,1}$ satisfies properties (G$_1$) and (G$_2$).
Simple calculations
establish that
\begin{align*} 
\sigma(L_{G_{1,1}}) 
	& = \left[-1, -\frac{1}{3}\right] \cup \left[\frac{1}{3},1\right]
\end{align*}
and that $L_G$ has no eigenvalues. 
In particular, $L_{G_{1,1}}$ has a spectral gap.
 
\begin{center}
\unitlength 0.1in
\begin{picture}( 33.5400, 10.6000)(  2.8000,-15.6000)
%
\special{pn 20}%
\special{sh 0}%
\special{ar 956 1462 124 98  0.0000000 6.2831853}%
%
\special{pn 20}%
\special{pa 960 720}%
\special{pa 960 1364}%
\special{fp}%
%
\special{pn 20}%
\special{pa 880 1430}%
\special{pa 1472 1430}%
\special{fp}%
%
\special{pn 8}%
\special{pa 952 1444}%
\special{pa 952 1444}%
\special{fp}%
%
\special{pn 20}%
\special{ar 956 1462 124 98  6.1795578 6.2831853}%
\special{ar 956 1462 124 98  0.0000000 4.0499627}%
%
\special{pn 20}%
\special{pa 960 720}%
\special{pa 960 1364}%
\special{fp}%
%
\special{pn 20}%
\special{sh 0}%
\special{ar 956 1462 124 98  0.0000000 6.2831853}%
%
\special{pn 20}%
\special{pa 960 720}%
\special{pa 960 1364}%
\special{fp}%
%
\special{pn 20}%
\special{pa 2208 704}%
\special{pa 2208 1348}%
\special{fp}%
%
\special{pn 20}%
\special{pa 2208 704}%
\special{pa 2208 1348}%
\special{fp}%
%
\special{pn 20}%
\special{pa 2208 704}%
\special{pa 2208 1348}%
\special{fp}%
%
\special{pn 20}%
\special{sh 0}%
\special{ar 2828 1440 124 100  0.0000000 6.2831853}%
%
\special{pn 20}%
\special{pa 2832 696}%
\special{pa 2832 1342}%
\special{fp}%
%
\special{pn 20}%
\special{pa 2826 1430}%
\special{pa 3364 1430}%
\special{fp}%
%
\special{pn 8}%
\special{pa 2826 1422}%
\special{pa 2826 1422}%
\special{fp}%
%
\special{pn 20}%
\special{ar 2828 1440 124 100  6.1821967 6.2831853}%
\special{ar 2828 1440 124 100  0.0000000 4.0450828}%
%
\special{pn 20}%
\special{pa 2832 696}%
\special{pa 2832 1342}%
\special{fp}%
%
\special{pn 20}%
\special{sh 0}%
\special{ar 2828 1440 124 100  0.0000000 6.2831853}%
%
\special{pn 20}%
\special{pa 2832 696}%
\special{pa 2832 1342}%
\special{fp}%
%
\special{pn 20}%
\special{pa 520 1440}%
\special{pa 822 1440}%
\special{fp}%
%
\special{pn 8}%
\special{pa 3370 1420}%
\special{pa 3634 1420}%
\special{dt 0.045}%
%
\special{pn 8}%
\special{pa 280 1440}%
\special{pa 570 1440}%
\special{dt 0.045}%
%
\special{pn 20}%
\special{pa 956 510}%
\special{pa 1066 710}%
\special{fp}%
\special{pa 1064 714}%
\special{pa 840 714}%
\special{fp}%
\special{pa 840 714}%
\special{pa 956 514}%
\special{fp}%
%
\special{pn 20}%
\special{pa 2838 500}%
\special{pa 2948 700}%
\special{fp}%
\special{pa 2944 704}%
\special{pa 2722 704}%
\special{fp}%
\special{pa 2722 704}%
\special{pa 2838 504}%
\special{fp}%
%
\special{pn 20}%
\special{pa 956 510}%
\special{pa 1066 710}%
\special{fp}%
\special{pa 1064 714}%
\special{pa 840 714}%
\special{fp}%
\special{pa 840 714}%
\special{pa 956 514}%
\special{fp}%
%
\special{pn 20}%
\special{pa 2204 510}%
\special{pa 2314 710}%
\special{fp}%
\special{pa 2312 714}%
\special{pa 2088 714}%
\special{fp}%
\special{pa 2088 714}%
\special{pa 2204 514}%
\special{fp}%
%
\special{pn 20}%
\special{pa 1568 704}%
\special{pa 1568 1348}%
\special{fp}%
%
\special{pn 20}%
\special{pa 1580 1440}%
\special{pa 2060 1440}%
\special{fp}%
%
\special{pn 8}%
\special{pa 1560 1428}%
\special{pa 1560 1428}%
\special{fp}%
%
\special{pn 20}%
\special{pa 1568 704}%
\special{pa 1568 1348}%
\special{fp}%
%
\special{pn 20}%
\special{pa 1568 704}%
\special{pa 1568 1348}%
\special{fp}%
%
\special{pn 20}%
\special{pa 1564 510}%
\special{pa 1674 710}%
\special{fp}%
\special{pa 1672 714}%
\special{pa 1448 714}%
\special{fp}%
\special{pa 1448 714}%
\special{pa 1564 514}%
\special{fp}%
%
\special{pn 20}%
\special{sh 0}%
\special{ar 1570 1450 124 98  0.0000000 6.2831853}%
%
\special{pn 20}%
\special{sh 0}%
\special{ar 2200 1450 124 98  0.0000000 6.2831853}%
%
\special{pn 20}%
\special{pa 2340 1430}%
\special{pa 2710 1430}%
\special{fp}%
\end{picture}%
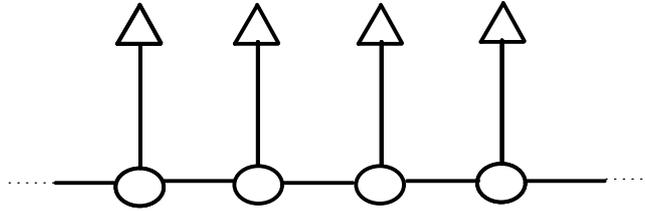
\figcaption{Graph $G_{1,1}$.} \label{fig00}
\end{center}
\noindent
In what follows, we consider a graph $G_{2,1}$ 
whose pendant edges are connected to alternate vertices of $\mathbb{Z}$:
$V(G_{2,1}) = \{ (n,0) \mid n \in \mathbb{Z} \} \cup \{(2n+1,1) \mid n \in \mathbb{Z} \}$  (see Figure \ref{fig01}).
$G_{2,1}$ satisfies (G$_1$) and (G$_2$). 
Simple calculations show that
\[ \sigma(L_{G_{2,1}}) = \left[-1, -\frac{1}{\sqrt 3}\right] \cup \{ 0 \} \cup \left[\frac{1}{\sqrt 3}, 1\right] 
	\]
and that $0$ is an eigenvalue of $L_{G_{2,1}}$.
Hence, $L_{G_{2,1}}$ has a spectral gap and a zero eigenvalue.
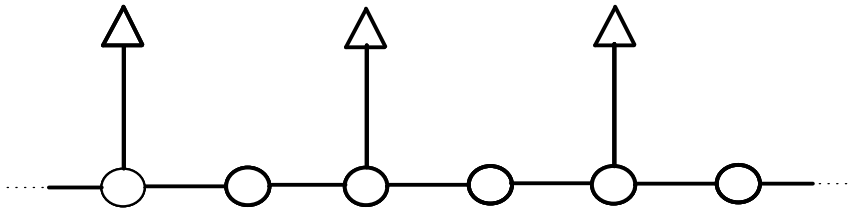
\begin{center}
\unitlength 0.1in
\begin{picture}( 44.2000, 10.4000)(  2.9000,-15.5000)
%
\special{pn 13}%
\special{sh 0}%
\special{ar 890 1454 112 98  0.0000000 6.2831853}%
%
\special{pn 20}%
\special{pa 894 718}%
\special{pa 894 1356}%
\special{fp}%
%
\special{pn 20}%
\special{sh 0}%
\special{ar 1536 1448 112 98  0.0000000 6.2831853}%
%
\special{pn 20}%
\special{pa 888 1448}%
\special{pa 1422 1448}%
\special{fp}%
%
\special{pn 8}%
\special{pa 888 1434}%
\special{pa 888 1434}%
\special{fp}%
%
\special{pn 13}%
\special{ar 890 1454 112 98  6.1801541 6.2831853}%
\special{ar 890 1454 112 98  0.0000000 4.0459237}%
%
\special{pn 20}%
\special{pa 894 718}%
\special{pa 894 1356}%
\special{fp}%
%
\special{pn 20}%
\special{sh 0}%
\special{ar 1536 1448 112 98  0.0000000 6.2831853}%
%
\special{pn 13}%
\special{sh 0}%
\special{ar 890 1454 112 98  0.0000000 6.2831853}%
%
\special{pn 20}%
\special{pa 894 718}%
\special{pa 894 1356}%
\special{fp}%
%
\special{pn 20}%
\special{sh 0}%
\special{ar 1536 1448 112 98  0.0000000 6.2831853}%
%
\special{pn 20}%
\special{sh 0}%
\special{ar 2148 1448 110 98  0.0000000 6.2831853}%
%
\special{pn 20}%
\special{pa 2152 712}%
\special{pa 2152 1350}%
\special{fp}%
%
\special{pn 20}%
\special{ar 2792 1440 112 98  0.0000000 6.2831853}%
%
\special{pn 20}%
\special{pa 2146 1440}%
\special{pa 2678 1440}%
\special{fp}%
%
\special{pn 8}%
\special{pa 2146 1428}%
\special{pa 2146 1428}%
\special{fp}%
%
\special{pn 20}%
\special{ar 2148 1448 110 98  6.1814666 6.2831853}%
\special{ar 2148 1448 110 98  0.0000000 4.0444599}%
%
\special{pn 20}%
\special{pa 2152 712}%
\special{pa 2152 1350}%
\special{fp}%
%
\special{pn 20}%
\special{ar 2792 1440 112 98  0.0000000 6.2831853}%
%
\special{pn 20}%
\special{sh 0}%
\special{ar 2148 1448 110 98  0.0000000 6.2831853}%
%
\special{pn 20}%
\special{pa 2152 712}%
\special{pa 2152 1350}%
\special{fp}%
%
\special{pn 20}%
\special{ar 2792 1440 112 98  0.0000000 6.2831853}%
%
\special{pn 20}%
\special{sh 0}%
\special{ar 3430 1440 112 98  0.0000000 6.2831853}%
%
\special{pn 20}%
\special{pa 3434 704}%
\special{pa 3434 1344}%
\special{fp}%
%
\special{pn 20}%
\special{ar 4076 1434 112 98  0.0000000 6.2831853}%
%
\special{pn 20}%
\special{pa 3428 1434}%
\special{pa 3960 1434}%
\special{fp}%
%
\special{pn 8}%
\special{pa 3428 1422}%
\special{pa 3428 1422}%
\special{fp}%
%
\special{pn 20}%
\special{ar 3430 1440 112 98  6.1814666 6.2831853}%
\special{ar 3430 1440 112 98  0.0000000 4.0390198}%
%
\special{pn 20}%
\special{pa 3434 704}%
\special{pa 3434 1344}%
\special{fp}%
%
\special{pn 20}%
\special{ar 4076 1434 112 98  0.0000000 6.2831853}%
%
\special{pn 20}%
\special{sh 0}%
\special{ar 3430 1440 112 98  0.0000000 6.2831853}%
%
\special{pn 20}%
\special{pa 3434 704}%
\special{pa 3434 1344}%
\special{fp}%
%
\special{pn 20}%
\special{ar 4076 1434 112 98  0.0000000 6.2831853}%
%
\special{pn 20}%
\special{pa 1656 1440}%
\special{pa 2042 1440}%
\special{fp}%
%
\special{pn 20}%
\special{pa 2898 1432}%
\special{pa 3316 1432}%
\special{fp}%
%
\special{pn 20}%
\special{pa 4196 1434}%
\special{pa 4462 1434}%
\special{fp}%
%
\special{pn 20}%
\special{pa 506 1454}%
\special{pa 780 1454}%
\special{fp}%
%
\special{pn 8}%
\special{pa 4450 1434}%
\special{pa 4710 1434}%
\special{dt 0.045}%
%
\special{pn 8}%
\special{pa 290 1454}%
\special{pa 550 1454}%
\special{dt 0.045}%
%
\special{pn 20}%
\special{pa 892 510}%
\special{pa 990 708}%
\special{fp}%
\special{pa 988 712}%
\special{pa 786 712}%
\special{fp}%
\special{pa 786 712}%
\special{pa 892 514}%
\special{fp}%
%
\special{pn 20}%
\special{pa 3438 510}%
\special{pa 3538 708}%
\special{fp}%
\special{pa 3536 712}%
\special{pa 3334 712}%
\special{fp}%
\special{pa 3334 712}%
\special{pa 3438 514}%
\special{fp}%
%
\special{pn 20}%
\special{pa 892 510}%
\special{pa 990 708}%
\special{fp}%
\special{pa 988 712}%
\special{pa 786 712}%
\special{fp}%
\special{pa 786 712}%
\special{pa 892 514}%
\special{fp}%
%
\special{pn 20}%
\special{pa 2148 520}%
\special{pa 2248 718}%
\special{fp}%
\special{pa 2244 722}%
\special{pa 2044 722}%
\special{fp}%
\special{pa 2044 722}%
\special{pa 2148 524}%
\special{fp}%
\end{picture}%
\figcaption{Graph $G_{2,1}$.} \label{fig01}
\end{center}

We note that graphs $G$ such as $G_{1,0} \equiv \mathbb{Z}$, $G_{1,1}$ and $G_{2,1}$ have the following periodicity:
\begin{itemize}
\item[(G$_3$)] There exists a number $r \in \mathbb{N}$ such that
\begin{equation} 
\label{02/24/23:02}
\mbox{$(n+r,1) \in V(G)$ if and only if $(n,1) \in V(G)$.}
\end{equation}
\end{itemize}
We use $\mathscr{G}$ to denote the set of graphs
that satisfy conditions (G$_1$) - (G$_3$).
Each graph $G \in \mathscr{G}$ is a covering graph.
The spectral properties of the Laplacian on general covering graphs
are reported in \cite{HiNo}.
In this paper, we focus on the covering graphs $G \in \mathscr{G}$ and investigate the spectrum of $L_G$ in detail.
Give an $n \in \mathbb{N}$ and
$r$ satisfying \eqref{02/24/23:02}, 
we define the following subset of $V(G)$, which we call a {\it cell} of $G$:
\[ V_{n,r}(G) = \{ (n + m,s) \in V(G) \mid m = 0, 1, \cdots, r-1, s=0,1 \}. 
	\] 
We note that the following relationships hold:
\[ V(G) = \bigcup_{n \in \mathbb{Z}} V_{nr,r}(G), \quad V_{nr,r}(G) \cap V_{mr,r}(G) = \emptyset \quad (n \not=m). \]
We use $s(r)$ to denote the number of end vertices in the cell $V_{1,r}(G)$ 
:
\begin{align*} 
& s(r) = \# \{(n,1) \in V(G) \mid n=1, 2, \cdots, r \}.
\end{align*}
For instance, we have $r(G_{1,0}) = 1$, $s(G_{1,0}) = 0$; $r(G_{1,1}) = 1$, $s(G_{1,1}) = 1$ 
and $r(G_{2,1}) = 2$, $s(G_{2,1}) = 1$.

We say that $G$ and $H$ are isomorphic and write this as $G \simeq H$ 
if there exists a bijection $\phi:V(G) \longrightarrow V(H)$ 
such that $\phi(x) \sim \phi(y)$ if and only if $x \sim y$. 
For any graphs $G$, $H \in \mathscr{G}$,
$G \simeq H$ holds if and only if there exists a number $r$ satisfying \eqref{02/24/23:02} 
for both $G$ and $H$ such that $V_{1,r}(G) = V_{n,r}(H)$
for some $n \in \mathbb{Z}$. 
If $G \simeq H$, then $\sigma(L_G) = \sigma(L_H)$.

A graph $G  \in \mathscr{G}$ is uniquely determined (up to isomorphism) by 
the number $s(r)$ of $(n,1) \in V_{1,r}(G)$ and the arrangement of $(n,1) \in V_{1,r}(G)$.
In particular, 
the graph $G \in \mathscr{G}$ with $s(1)=1$ has a single end vertex $(1,1) \in V(G)$ in the cell $V_{1,1}(G)$,
and hence $G = G_{1,1}$.
A graph $G \in \mathscr{G}$ with $s(2)=1$ has a single end vertex $(n,1) \in V(G)$ 
in the cell $V_{1,2}(G)$; hence, $G$ is isomorphic to $G_{2,1}$.
Similarly, for any $r \in \mathbb{N}$ 
we can define a graph $G_{r,1} \in \mathscr{G}$ that has a single end vertex $(n,1) \in V(G)$ in the cell $V_{1,r}(G)$. 
Clearly, $G_{r,1}$ is uniquely determined up to isomorphism.
For any $r \in \mathbb{N}$, we can define a graph $G_{r,{r-1}}$ such that $s(r) = r-1$, 
{\it i.e.} the number of end vertices $(n,1) \in V(G)$ in the cell $V_{1,r}(G)$ is equal to $r-1$.

On the other hand, in the case where $s(r) \not= 1$ or $s(r) \not= r-1$, 
there are a variety of graphs that are not isomorphic to each other.
In general, for two graphs $G$ and $\tilde{G}$ which are not isomorphic to each other, 
we cannot expect that $\sigma(L_{G}) = \sigma(L_{\tilde{G}})$.
For instance, assuming $r=4$ and $s(r)=2$, we have two different graphs as follows:
1) if the end vertices $(n,1) \in V_{1,4}(G_{4,2})$ are connected to alternate vertices, 
then the graph $G_{4,2} \in \mathscr{G}$ is isomorphic to $G_{2,1}$,
and
2) the graph $\tilde{G}_{4,2} \in \mathscr{G}$ with end vertices $(1,1)$ and $(2,1) \in V_{1,4}(\tilde{G}_{4,2})$ 
satisfies $s(4)=2$ but is not isomorphic to $G_{2,1}$. 
See Figures \ref{figG42} and \ref{figG42tilde}.
\begin{center}
\unitlength 0.1in
\begin{picture}( 33.5400, 12.8000)(  2.8000,-16.7000)
%
\special{pn 20}%
\special{sh 0}%
\special{ar 956 1462 124 98  0.0000000 6.2831853}%
%
\special{pn 20}%
\special{pa 960 720}%
\special{pa 960 1364}%
\special{fp}%
%
\special{pn 20}%
\special{pa 880 1430}%
\special{pa 1472 1430}%
\special{fp}%
%
\special{pn 8}%
\special{pa 952 1444}%
\special{pa 952 1444}%
\special{fp}%
%
\special{pn 20}%
\special{ar 956 1462 124 98  6.1795578 6.2831853}%
\special{ar 956 1462 124 98  0.0000000 4.0499627}%
%
\special{pn 20}%
\special{pa 960 720}%
\special{pa 960 1364}%
\special{fp}%
%
\special{pn 20}%
\special{sh 0}%
\special{ar 956 1462 124 98  0.0000000 6.2831853}%
%
\special{pn 20}%
\special{pa 960 720}%
\special{pa 960 1364}%
\special{fp}%
%
\special{pn 20}%
\special{pa 2208 704}%
\special{pa 2208 1348}%
\special{fp}%
%
\special{pn 20}%
\special{pa 2208 704}%
\special{pa 2208 1348}%
\special{fp}%
%
\special{pn 20}%
\special{pa 2208 704}%
\special{pa 2208 1348}%
\special{fp}%
%
\special{pn 20}%
\special{sh 0}%
\special{ar 2828 1440 124 100  0.0000000 6.2831853}%
%
\special{pn 20}%
\special{pa 2826 1430}%
\special{pa 3364 1430}%
\special{fp}%
%
\special{pn 8}%
\special{pa 2826 1422}%
\special{pa 2826 1422}%
\special{fp}%
%
\special{pn 20}%
\special{ar 2828 1440 124 100  6.1821967 6.2831853}%
\special{ar 2828 1440 124 100  0.0000000 4.0450828}%
%
\special{pn 20}%
\special{sh 0}%
\special{ar 2828 1440 124 100  0.0000000 6.2831853}%
%
\special{pn 20}%
\special{pa 520 1440}%
\special{pa 822 1440}%
\special{fp}%
%
\special{pn 8}%
\special{pa 3370 1420}%
\special{pa 3634 1420}%
\special{dt 0.045}%
%
\special{pn 8}%
\special{pa 280 1440}%
\special{pa 570 1440}%
\special{dt 0.045}%
%
\special{pn 20}%
\special{pa 956 510}%
\special{pa 1066 710}%
\special{fp}%
\special{pa 1064 714}%
\special{pa 840 714}%
\special{fp}%
\special{pa 840 714}%
\special{pa 956 514}%
\special{fp}%
%
\special{pn 20}%
\special{pa 956 510}%
\special{pa 1066 710}%
\special{fp}%
\special{pa 1064 714}%
\special{pa 840 714}%
\special{fp}%
\special{pa 840 714}%
\special{pa 956 514}%
\special{fp}%
%
\special{pn 20}%
\special{pa 2204 510}%
\special{pa 2314 710}%
\special{fp}%
\special{pa 2312 714}%
\special{pa 2088 714}%
\special{fp}%
\special{pa 2088 714}%
\special{pa 2204 514}%
\special{fp}%
%
\special{pn 20}%
\special{pa 1580 1440}%
\special{pa 2060 1440}%
\special{fp}%
%
\special{pn 8}%
\special{pa 1560 1428}%
\special{pa 1560 1428}%
\special{fp}%
%
\special{pn 20}%
\special{sh 0}%
\special{ar 1570 1450 124 98  0.0000000 6.2831853}%
%
\special{pn 20}%
\special{sh 0}%
\special{ar 2200 1450 124 98  0.0000000 6.2831853}%
%
\special{pn 20}%
\special{pa 2340 1430}%
\special{pa 2710 1430}%
\special{fp}%
%
\special{pn 8}%
\special{pa 570 390}%
\special{pa 3190 390}%
\special{pa 3190 1670}%
\special{pa 570 1670}%
\special{pa 570 390}%
\special{fp}%
\put(36.0000,-10.0000){\makebox(0,0){$V_{1,4}(G_{4,2})$}}%
\end{picture}%
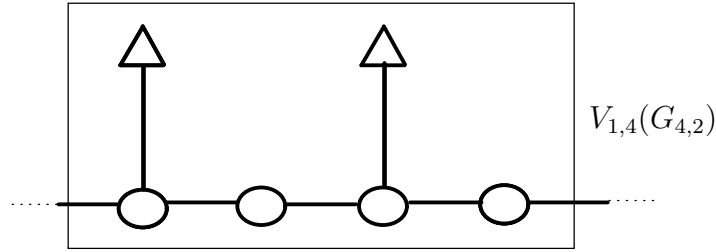
\figcaption{Graph $G_{4,2} \in \mathscr{G}$ has end vertices $(1,1)$ 
and $(3,1) \in V_{1,4}(G_{4,2})$.} \label{figG42}
\end{center}
\begin{center}
\unitlength 0.1in
\begin{picture}( 33.5400, 12.8000)(  2.8000,-16.7000)
%
\special{pn 20}%
\special{sh 0}%
\special{ar 956 1462 124 98  0.0000000 6.2831853}%
%
\special{pn 20}%
\special{pa 960 720}%
\special{pa 960 1364}%
\special{fp}%
%
\special{pn 20}%
\special{pa 880 1430}%
\special{pa 1472 1430}%
\special{fp}%
%
\special{pn 8}%
\special{pa 952 1444}%
\special{pa 952 1444}%
\special{fp}%
%
\special{pn 20}%
\special{ar 956 1462 124 98  6.1795578 6.2831853}%
\special{ar 956 1462 124 98  0.0000000 4.0499627}%
%
\special{pn 20}%
\special{pa 960 720}%
\special{pa 960 1364}%
\special{fp}%
%
\special{pn 20}%
\special{sh 0}%
\special{ar 956 1462 124 98  0.0000000 6.2831853}%
%
\special{pn 20}%
\special{pa 960 720}%
\special{pa 960 1364}%
\special{fp}%
%
\special{pn 20}%
\special{sh 0}%
\special{ar 2828 1440 124 100  0.0000000 6.2831853}%
%
\special{pn 20}%
\special{pa 2832 696}%
\special{pa 2832 1342}%
\special{fp}%
%
\special{pn 20}%
\special{pa 2826 1430}%
\special{pa 3364 1430}%
\special{fp}%
%
\special{pn 8}%
\special{pa 2826 1422}%
\special{pa 2826 1422}%
\special{fp}%
%
\special{pn 20}%
\special{ar 2828 1440 124 100  6.1821967 6.2831853}%
\special{ar 2828 1440 124 100  0.0000000 4.0450828}%
%
\special{pn 20}%
\special{pa 2832 696}%
\special{pa 2832 1342}%
\special{fp}%
%
\special{pn 20}%
\special{sh 0}%
\special{ar 2828 1440 124 100  0.0000000 6.2831853}%
%
\special{pn 20}%
\special{pa 2832 696}%
\special{pa 2832 1342}%
\special{fp}%
%
\special{pn 20}%
\special{pa 520 1440}%
\special{pa 822 1440}%
\special{fp}%
%
\special{pn 8}%
\special{pa 3370 1420}%
\special{pa 3634 1420}%
\special{dt 0.045}%
%
\special{pn 8}%
\special{pa 280 1440}%
\special{pa 570 1440}%
\special{dt 0.045}%
%
\special{pn 20}%
\special{pa 956 510}%
\special{pa 1066 710}%
\special{fp}%
\special{pa 1064 714}%
\special{pa 840 714}%
\special{fp}%
\special{pa 840 714}%
\special{pa 956 514}%
\special{fp}%
%
\special{pn 20}%
\special{pa 2838 500}%
\special{pa 2948 700}%
\special{fp}%
\special{pa 2944 704}%
\special{pa 2722 704}%
\special{fp}%
\special{pa 2722 704}%
\special{pa 2838 504}%
\special{fp}%
%
\special{pn 20}%
\special{pa 956 510}%
\special{pa 1066 710}%
\special{fp}%
\special{pa 1064 714}%
\special{pa 840 714}%
\special{fp}%
\special{pa 840 714}%
\special{pa 956 514}%
\special{fp}%
%
\special{pn 20}%
\special{pa 1580 1440}%
\special{pa 2060 1440}%
\special{fp}%
%
\special{pn 8}%
\special{pa 1560 1428}%
\special{pa 1560 1428}%
\special{fp}%
%
\special{pn 20}%
\special{sh 0}%
\special{ar 1570 1450 124 98  0.0000000 6.2831853}%
%
\special{pn 20}%
\special{sh 0}%
\special{ar 2200 1450 124 98  0.0000000 6.2831853}%
%
\special{pn 20}%
\special{pa 2340 1430}%
\special{pa 2710 1430}%
\special{fp}%
%
\special{pn 8}%
\special{pa 3190 1670}%
\special{pa 570 1670}%
\special{pa 570 390}%
\special{pa 3190 390}%
\special{pa 3190 1670}%
\special{fp}%
\put(36.0000,-10.0000){\makebox(0,0){$V_{1,4}(\tilde{G}_{4,2})$}}%
\end{picture}%
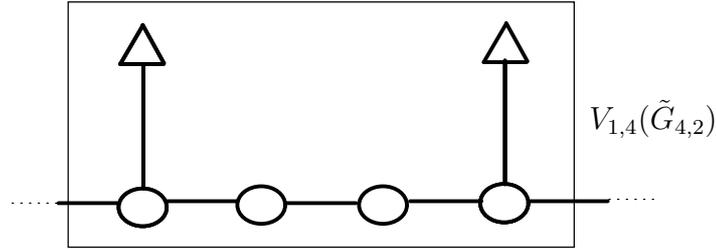
\figcaption{Graph $\tilde{G}_{4,2} \in \mathscr{G}$ has end vertices $(1,1)$ 
and $(2,1) \in V_{1,4}(\tilde{G}_{4,2})$.} \label{figG42tilde}
\end{center}
As we shall see later, the spectra of $L_{G_{2,1}}$ and $L_{G_{4,2}}$ are different.
In particular, $L_{G_{2,1}}$ has zero as an eigenvalue 
whereas $L_{G_{4,2}}$ has no eigenvalue.

\subsection{Results}
Let $\mathscr{G}_\times = \mathscr{G} \setminus \{ \mathbb{Z}\}$.
We first identify spectral properties of $L_G$ that hold for all graphs 
$G \in \mathscr{G}_\times$.
\begin{theorem}
\label{24/02/29/17:13}
{\rm
Let $G \in \mathscr{G}_\times$. Then the following hold:
\begin{itemize}
\item[(1)] $\sigma(L_G) \subset [-1,1]$ is symmetric with respect to zero.
\item[(2)] There exists a constant $\epsilon > 0$ such that 
\[ (-\epsilon, 0) \cup (0, \epsilon) \subset \rho(L_G), 
	 \]
where $\rho(L_G)$ denotes the resolvent set of $L_G$.
\item[(3)] There is no eigenvalue in $\sigma(L_G) \setminus\{0\}$.
\end{itemize}
}
\end{theorem}
We will prove Theorem \ref{24/02/29/17:13} in Section \ref{24/04/14/22:54};
here we will make two comments about the theorem:
(i) The symmetry of $\sigma(L_G)$ with respect to zero in Theorem \ref{24/02/29/17:13} (1) 
comes from the bipartiteness of $G \in \mathscr{G}$. 
If we remove condition (G$_2$) from the assumptions of Theorem \ref{24/02/29/17:13},
(1) will not hold true in general, as shown in Figure \ref{lattce}.
\begin{center}
\unitlength 0.1in
\begin{picture}( 33.5400,  6.6000)(  2.8000,-15.6000)
%
\special{pn 20}%
\special{sh 0}%
\special{ar 956 1462 124 98  0.0000000 6.2831853}%
%
\special{pn 20}%
\special{pa 880 1430}%
\special{pa 1472 1430}%
\special{fp}%
%
\special{pn 8}%
\special{pa 952 1444}%
\special{pa 952 1444}%
\special{fp}%
%
\special{pn 20}%
\special{ar 956 1462 124 98  6.1795578 6.2831853}%
\special{ar 956 1462 124 98  0.0000000 4.0499627}%
%
\special{pn 20}%
\special{sh 0}%
\special{ar 956 1462 124 98  0.0000000 6.2831853}%
%
\special{pn 20}%
\special{sh 0}%
\special{ar 2828 1440 124 100  0.0000000 6.2831853}%
%
\special{pn 20}%
\special{pa 2826 1430}%
\special{pa 3364 1430}%
\special{fp}%
%
\special{pn 8}%
\special{pa 2826 1422}%
\special{pa 2826 1422}%
\special{fp}%
%
\special{pn 20}%
\special{ar 2828 1440 124 100  6.1821967 6.2831853}%
\special{ar 2828 1440 124 100  0.0000000 4.0450828}%
%
\special{pn 20}%
\special{sh 0}%
\special{ar 2828 1440 124 100  0.0000000 6.2831853}%
%
\special{pn 20}%
\special{pa 520 1440}%
\special{pa 822 1440}%
\special{fp}%
%
\special{pn 8}%
\special{pa 3370 1420}%
\special{pa 3634 1420}%
\special{dt 0.045}%
%
\special{pn 8}%
\special{pa 280 1440}%
\special{pa 570 1440}%
\special{dt 0.045}%
%
\special{pn 20}%
\special{pa 1580 1440}%
\special{pa 2060 1440}%
\special{fp}%
%
\special{pn 8}%
\special{pa 1560 1428}%
\special{pa 1560 1428}%
\special{fp}%
%
\special{pn 20}%
\special{sh 0}%
\special{ar 1570 1450 124 98  0.0000000 6.2831853}%
%
\special{pn 20}%
\special{sh 0}%
\special{ar 2200 1450 124 98  0.0000000 6.2831853}%
%
\special{pn 20}%
\special{pa 2340 1430}%
\special{pa 2710 1430}%
\special{fp}%
%
\special{pn 20}%
\special{sh 0}%
\special{ar 1244 998 124 98  0.0000000 6.2831853}%
%
\special{pn 20}%
\special{ar 1244 998 124 98  6.1795578 6.2831853}%
\special{ar 1244 998 124 98  0.0000000 4.0499627}%
%
\special{pn 20}%
\special{sh 0}%
\special{ar 1244 998 124 98  0.0000000 6.2831853}%
%
\special{pn 20}%
\special{sh 0}%
\special{ar 2514 998 124 98  0.0000000 6.2831853}%
%
\special{pn 20}%
\special{ar 2514 998 124 98  6.1795578 6.2831853}%
\special{ar 2514 998 124 98  0.0000000 4.0499627}%
%
\special{pn 20}%
\special{sh 0}%
\special{ar 2514 998 124 98  0.0000000 6.2831853}%
%
\special{pn 20}%
\special{pa 1000 1360}%
\special{pa 1160 1080}%
\special{fp}%
\special{pa 1340 1080}%
\special{pa 1510 1340}%
\special{fp}%
\special{pa 2260 1350}%
\special{pa 2430 1090}%
\special{fp}%
\special{pa 2590 1080}%
\special{pa 2770 1340}%
\special{fp}%
\end{picture}%
\figcaption{
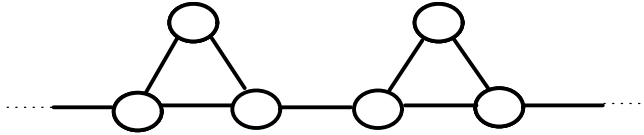
Let $G$ be a graph satisfying (G$_1$) and (G$_3$) with $r=2$ 
and suppose that $s(2) = 1$, ${\rm deg}(n,0) =3$, 
and $(1,1) \in V_{1,2}(G)$ is connected to vertices $(1,0)$ and $(2,0)$.
Then (G$_2$) does not hold. 
We observe that $1 \in \sigma(L_G)$, but $-1 \not\in \sigma(L_G)$;
hence, Theorem \ref{24/02/29/17:13} (1) does not hold. }
\label{lattce}
\end{center}
(ii) Theorem \ref{24/02/29/17:13} (2) implies that 
for all graphs $G \in \mathscr{G}_\times$, 
$L_G$ has a spectral gap around zero.
The situation is different from the two-dimensional case.
In Section \ref{twodim},
we will prove the following for graphs obtained from $\mathbb{Z}^2$ by adding pendant edges: 
\begin{itemize}
\item[(a)] there is an arrangement of pendant edges 
such that the Laplacian has no spectral gap (Theorem \ref{05/11/16:46} (a));
\item[(b)] there is an arrangement of pendant edges 
such that the Laplacian has a spectral gap (Theorem \ref{05/11/16:46} (b)).
\end{itemize}
In a companion paper \cite{TS}, 
we study graphs obtained from the hexagonal lattice by adding pendant edges
and establish results similar to that mentioned above.
We believe that the spectra of the Laplacians on such graphs are related 
to the electronic structure of hydrogenated graphenes and graphane;
see \cite{graphene} for details.

Subsequently, we introduce a decomposition of the cells of a graph $G$
in order to formulate the condition under which $L_G$ has an eigenvalue.
From Theorem \ref{24/02/29/17:13} (3), we know that
only the zero eigenvalue can exist:
if it exists, then it is an isolated eigenvalue of $L_G$. 
For any $r$ satisfying \eqref{02/24/23:02},
we define 
\[ U_{r}(G) = \{ (n,0) \in V_{1,r}(G) \mid (n,1) \not\in V_{1,r}(G) \}. \]
Without loss of generality,
we can assume that $(1,1) \in V_{1,r}(G)$
because there exists a graph $\tilde{G}$ isomorphic to $G$ such that $(1,1) \in V_{1,r}(\tilde{G})$.
Letting $\tilde{p}(r) = \#  U_r(G)$,
we can write
\[ U_r(G) = \{ (n_1,0), (n_2,0), \cdots, (n_{\tilde{p}(r)},0) \} \]
with $1 < n_1 < n_2 < \cdots < n_{\tilde{p}(r)}$.
Then there exist a $p(r)$ with $1 \leq p(r) \leq \tilde{p}(r)$ and $l_1, l_2, \cdots, l_{p(r)} \geq 1$ such that 
$U_r(G)$ can be decomposed into the connected subgraphs 
$U_i$ ($i=1,2,\cdots,p(r)$):
\begin{equation} 
\label{24/03/01/15:53}
U_r(G) = \bigcup_{i=1}^{p(r)} U_i, 
\end{equation}
where 
\begin{equation}
\label{24/03/24/15:56} 
U_i = \left\{ \left(m_{1+ \sum_{j=1}^{i-1} l_j},0 \right), \left(m_{2+ \sum_{j=1}^{i-1} l_j},0 \right),
	\cdots, \left(m_{\sum_{j=1}^{i} l_j},0 \right) \right\} 
\end{equation}
with 
\[ m_{t+1} - m_{t} = 1, \quad t= 1+ \sum_{j=1}^{i-1} l_j, 2 + \sum_{j=1}^{i-1} l_j, \cdots,l_i - 1 + \sum_{j=1}^{i-1} l_j \]
\[ m_{t+1} - m_{t} > 1, \quad t = \sum_{j=1}^{i} l_j, ~ i=1,2,\cdots,p(r)-1 \] 
and $\sum_{i= 1}^{p(r)} l_i = \tilde{p}(r)$. 
See Figure \ref{fig03} for an example. 
We use {\Large $\bullet$}, {\Large $\circ$} and {\small $\triangle$} 
to denote
the vertices $(n,0) \in U_r(G)$, $(n,0) \not\in U_r(G)$ and $(n,1) \in V_{1,r}(G)$, respectively.
Clearly, this decomposition is unique. 
\begin{center}
\unitlength 0.1in
\begin{picture}( 44.1000, 15.1700)(  2.9000,-18.7700)
%
\special{pn 13}%
\special{sh 0}%
\special{ar 888 1454 112 98  0.0000000 6.2831853}%
%
\special{pn 20}%
\special{pa 892 718}%
\special{pa 892 1356}%
\special{fp}%
%
\special{pn 20}%
\special{sh 0.600}%
\special{ar 1532 1448 112 98  0.0000000 6.2831853}%
%
\special{pn 20}%
\special{pa 886 1448}%
\special{pa 1418 1448}%
\special{fp}%
%
\special{pn 8}%
\special{pa 886 1434}%
\special{pa 886 1434}%
\special{fp}%
%
\special{pn 13}%
\special{ar 888 1454 112 98  6.1800028 6.2831853}%
\special{ar 888 1454 112 98  0.0000000 4.0583063}%
%
\special{pn 20}%
\special{pa 892 718}%
\special{pa 892 1356}%
\special{fp}%
%
\special{pn 20}%
\special{sh 0.600}%
\special{ar 1532 1448 112 98  0.0000000 6.2831853}%
%
\special{pn 13}%
\special{sh 0}%
\special{ar 888 1454 112 98  0.0000000 6.2831853}%
%
\special{pn 20}%
\special{pa 892 718}%
\special{pa 892 1356}%
\special{fp}%
%
\special{pn 20}%
\special{sh 0.600}%
\special{ar 1532 1448 112 98  0.0000000 6.2831853}%
%
\special{pn 20}%
\special{sh 0.600}%
\special{ar 2786 1440 112 98  0.0000000 6.2831853}%
%
\special{pn 20}%
\special{pa 2270 1450}%
\special{pa 2682 1450}%
\special{fp}%
%
\special{pn 8}%
\special{pa 2144 1426}%
\special{pa 2144 1426}%
\special{fp}%
%
\special{pn 20}%
\special{sh 0.600}%
\special{ar 2786 1440 112 98  0.0000000 6.2831853}%
%
\special{pn 20}%
\special{sh 0.600}%
\special{ar 2786 1440 112 98  0.0000000 6.2831853}%
%
\special{pn 20}%
\special{sh 0}%
\special{ar 3424 1440 112 98  0.0000000 6.2831853}%
%
\special{pn 20}%
\special{pa 3426 704}%
\special{pa 3426 1344}%
\special{fp}%
%
\special{pn 20}%
\special{sh 0.600}%
\special{ar 4066 1434 112 98  0.0000000 6.2831853}%
%
\special{pn 20}%
\special{pa 3420 1434}%
\special{pa 3952 1434}%
\special{fp}%
%
\special{pn 8}%
\special{pa 3420 1422}%
\special{pa 3420 1422}%
\special{fp}%
%
\special{pn 20}%
\special{ar 3424 1440 112 98  6.1827833 6.2831853}%
\special{ar 3424 1440 112 98  0.0000000 4.0335044}%
%
\special{pn 20}%
\special{pa 3426 704}%
\special{pa 3426 1344}%
\special{fp}%
%
\special{pn 20}%
\special{sh 0.600}%
\special{ar 4066 1434 112 98  0.0000000 6.2831853}%
%
\special{pn 20}%
\special{sh 0}%
\special{ar 3424 1440 112 98  0.0000000 6.2831853}%
%
\special{pn 20}%
\special{pa 3426 704}%
\special{pa 3426 1344}%
\special{fp}%
%
\special{pn 20}%
\special{sh 0.600}%
\special{ar 4066 1434 112 98  0.0000000 6.2831853}%
%
\special{pn 20}%
\special{pa 1652 1440}%
\special{pa 2038 1440}%
\special{fp}%
%
\special{pn 20}%
\special{pa 2892 1432}%
\special{pa 3310 1432}%
\special{fp}%
%
\special{pn 20}%
\special{pa 4188 1434}%
\special{pa 4454 1434}%
\special{fp}%
%
\special{pn 20}%
\special{pa 506 1454}%
\special{pa 778 1454}%
\special{fp}%
%
\special{pn 8}%
\special{pa 4440 1434}%
\special{pa 4700 1434}%
\special{dt 0.045}%
%
\special{pn 8}%
\special{pa 290 1454}%
\special{pa 550 1454}%
\special{dt 0.045}%
%
\special{pn 20}%
\special{pa 890 510}%
\special{pa 988 708}%
\special{fp}%
\special{pa 986 712}%
\special{pa 786 712}%
\special{fp}%
\special{pa 786 712}%
\special{pa 890 514}%
\special{fp}%
%
\special{pn 20}%
\special{pa 3432 510}%
\special{pa 3530 708}%
\special{fp}%
\special{pa 3528 712}%
\special{pa 3328 712}%
\special{fp}%
\special{pa 3328 712}%
\special{pa 3432 514}%
\special{fp}%
%
\special{pn 20}%
\special{pa 890 510}%
\special{pa 988 708}%
\special{fp}%
\special{pa 986 712}%
\special{pa 786 712}%
\special{fp}%
\special{pa 786 712}%
\special{pa 890 514}%
\special{fp}%
%
\special{pn 20}%
\special{sh 0.600}%
\special{ar 2150 1450 112 98  0.0000000 6.2831853}%
%
\special{pn 8}%
\special{pa 694 1564}%
\special{pa 694 1878}%
\special{fp}%
%
\special{pn 8}%
\special{pa 684 1874}%
\special{pa 4410 1874}%
\special{fp}%
\put(48.5000,-9.5000){\makebox(0,0){$V_{1,r}(G)$}}%
%
\special{pn 8}%
\special{pa 4306 1008}%
\special{pa 3854 1018}%
\special{dt 0.045}%
%
\special{pn 8}%
\special{pa 4082 1004}%
\special{pa 4082 768}%
\special{dt 0.045}%
\put(40.8100,-5.9500){\makebox(0,0){$U_2$}}%
%
\special{pn 8}%
\special{pa 2978 1020}%
\special{pa 1334 1028}%
\special{dt 0.045}%
%
\special{pn 8}%
\special{pa 2160 1016}%
\special{pa 2160 792}%
\special{dt 0.045}%
\put(21.5900,-6.2800){\makebox(0,0){$U_1$}}%
%
\special{pn 8}%
\special{pa 690 1536}%
\special{pa 690 370}%
\special{fp}%
\special{pa 690 370}%
\special{pa 4410 370}%
\special{fp}%
\special{pa 4400 360}%
\special{pa 4400 1878}%
\special{fp}%
%
\special{pn 8}%
\special{pa 1330 1030}%
\special{pa 1330 1670}%
\special{dt 0.045}%
\special{pa 3000 1680}%
\special{pa 3000 1030}%
\special{dt 0.045}%
%
\special{pn 8}%
\special{pa 1330 1680}%
\special{pa 3000 1680}%
\special{dt 0.045}%
%
\special{pn 8}%
\special{pa 3830 1010}%
\special{pa 3830 1680}%
\special{dt 0.045}%
\special{pa 4320 1680}%
\special{pa 4320 1020}%
\special{dt 0.045}%
%
\special{pn 8}%
\special{pa 3830 1680}%
\special{pa 4320 1680}%
\special{dt 0.045}%
\end{picture}%
\figcaption{Let $G \in \mathscr{G}$ satisfy \eqref{02/24/23:02} with $r=6$,
and let $(1,1), (5,1) \in V_{1,6}(G)$ be the end vertices of $G$.
Then the decomposition of $U_{6}(G)$ 
is $U_{6}(G) = U_1 \cup U_2$ with $U_1 = \{ (2,0), (3,0), (4,0)\}$ and $U_2 = \{ (6,0) \}$.} \label{fig03}
\end{center}

We have the following theorem:
 
\begin{theorem}
\label{24/03/01/16:03}
{\rm
Let $G \in \mathscr{G}_\times$ and let $U_r(G) = \bigcup_{i=1}^{p(r)} U_i$ be the decomposition defined in \eqref{24/03/01/15:53}.
Then zero is an eigenvalue of $L_G$
if and only if there exists an $i \in \{1,2, \cdots,p(r)\}$ such that $l_i = \#  U_i$ is odd.
}
\end{theorem}
We will prove Theorem \ref{24/03/01/16:03} in Subsection \ref{ss.3.3}.  

Because, by virtue of (G$_3$), $L_G$ is translation invariant for any $G \in \mathscr{G}$, 
we can produce more detailed information about the spectrum of $L_G$.
For an $r$ that satisfies \eqref{02/24/23:02},
we define the operator
$T_r : \ell^2(V(G)) \longrightarrow \ell^2(V(G))$ as
\[ (T_r\psi)(n,s) = \psi(\tau_r(n,s)), \quad \psi \in  \ell^2(V(G)), \]
where $\tau_r(n,s) = (n+r,s)$.
We can show that $L_G$ commutes with $T_r$, {\it i.e.} $[L_G, T_r]=0$
because for any $\psi \in  \ell^2(V(G))$,
\begin{align*} 
(T_r L_G \psi)(n,s) & = (L\psi)(\tau_r(n,s)) = \frac{1}{{\rm deg} (\tau_r(n,s))} \sum_{y \sim \tau_r(n,s)} \psi(y) \\
	& = \frac{1}{{\rm deg} (n,s)} \sum_{\tau_r^{-1}(y) \sim (n,s)} \psi(y) 
		= \frac{1}{{\rm deg} (n,s)} \sum_{\tilde{y} \sim (n,s)} \psi(\tau_r(\tilde{y})) \\
	& = (L_G T_r\psi)(n,s),
\end{align*}
where we have used the fact that ${\rm deg}(\tau_r(n,s)) = {\rm deg}(n,s)$ and set $\tilde y = \tau_r(y)$.

Let $X_r(G) = \{ (n,0) \in V_{1,r}(G) \mid (n,0) \not\in U_r(G) \}$ and set
\begin{equation}
\label{24/03/27/16:20} 
X_r(G) = \{ (i_1,0), (i_2,0), \cdots, (i_{s(r)},0) \}, 
	\quad i_1 < i_2 < \cdots < i_{s(r)}. 
\end{equation}
As mentioned above, we can assume that $i_1 = 1$ without loss of generality.
We define an identification $\iota:V(G) \longrightarrow \mathbb{Z} \times \{1,2,\cdots,r+s(r)\}$ with
\[ \iota(nr + j, s) 
	= \begin{cases} 
		(n,j) & \mbox{if $j=1, \cdots, r$ and $s=0$,} \\ 
		(n,r+p) & \mbox{if $j=i_p$, $p=1,2,\cdots,s(r)$ and $s=1$.} 
		\end{cases} \] 
for $(nr + j, s) \in V_{nr,r}(G)$,
and we define a unitary operator $\mathcal{J}:\ell^2(V(G)) \longrightarrow \ell^2(\mathbb{Z};\mathbb{C}^{r+s(r)})$ as
\begin{align*}
(\mathcal{J}\psi)(n) 
	= \begin{pmatrix} \psi_1(n) \\ 
			\mbox{\rotatebox{90}{$\cdots$}}\\ \psi_{r+s(r)}(n) 
		\end{pmatrix}
\end{align*}
with $\psi_i(n) = \psi(\iota^{-1}(n,i))$ ($i=1,2, \cdots,r+s(r)$).
The inner product of $\ell^2(\mathbb{Z};\mathbb{C}^{r+s(r)})$ is
\[ \langle \mathcal{J}\psi | \mathcal{J}\phi \rangle 
	:= \sum_{n \in \mathbb{Z}} \langle \mathcal{J}\psi(n), D \mathcal{J}\phi(n) \rangle_{\mathbb{C}^{r+s(r)}} \]
with $D = {\rm diag}(d_1, d_2, \cdots, d_{r+s(r)})$ and
\[ d_i = {\rm deg}(\iota^{-1}(n,i)), \quad i=1,\cdots, r + s(r). \]
Note that $d_i$ is independent of $n \in \mathbb{Z}$ 
because if $\iota^{-1}(n,i) = (nr+j,s)$, then
$\iota^{-1}(n+1,i) = ((n+1)r + j,s) = \tau_r(nr+j,s)$. 
Let $\mathcal{F}:\ell^2(\mathbb{Z};\mathbb{C}^{r+s(r)}) \longrightarrow L^2([-\pi,\pi];\mathbb{C}^{r+s(r)})$ 
be the discrete Fourier transformation:
\[ (\mathcal{F}\mathcal{J}\psi)(k) 
	= \begin{pmatrix} \widehat{\psi}_1(k) \\ 
			\mbox{\rotatebox{90}{$\cdots$}}\\ \widehat{\psi}_{r+s(r)}(k) 
		\end{pmatrix}, \quad k \in [-\pi,\pi], \]
where $\widehat{f}(k) = (2\pi)^{-1/2}\sum_{n \in \mathbb{Z}} e^{-ik n} f(n)$.
We define the quasi momentum operator $P$ as
\[ P = (\mathcal{F}\mathcal{J})^{-1} M_k (\mathcal{F}\mathcal{J}), \]
where $M_{g(k)}$ denotes the multiplication operator by a function $g(k)$ on $L^2([-\pi,\pi];\mathbb{C}^{r+s(r)})$. 
Clearly, we have $\sigma(P) = [-\pi,\pi]$.
We have the following lemma:
\begin{lemma}
\label{24/03/04/1:17}
{\rm 
Let $r$ satisfy \eqref{02/24/23:02}. Then $T_r = e^{iP}$. 
} 
\end{lemma}
\begin{proof}
Since we know from the functional calculus that
$e^{iP} = (\mathcal{F}\mathcal{J})^{-1} M_{e^{ik}} (\mathcal{F}\mathcal{J})$, 
it suffices to prove that $(\mathcal{F}\mathcal{J}) T_r =  M_{e^{ik}} (\mathcal{F}\mathcal{J})$.
Because
\[ e^{ik}\widehat{\psi}_i(k) = (2\pi)^{-1/2}\sum_{n \in \mathbb{Z}} e^{-ik (n-1)} \psi_i(n)
	= \widehat{\psi_i(\cdot + 1)}(k), \]
we have
\begin{align*}
(M_{e^{ik}}(\mathcal{F}\mathcal{J})\psi)(k) 
= \begin{pmatrix} \widehat{\psi_1(\cdot + 1)}(k) \\ 
			\mbox{\rotatebox{90}{$\cdots$}}\\ \widehat{\psi_{r+s(r)}(\cdot + 1)}(k) 
		\end{pmatrix} 
= (\mathcal{F}\mathcal{J}T_r\psi)(k),
\end{align*}
where we have used the following:
\begin{align*}
\psi_i(n+1) 
	& = \psi(\iota^{-1}(n+1,i)) \\
	& = \begin{cases}
		\psi((n+1)r+j,0), & i = j~ (j=1, \cdots, r) \\
		\psi((n+1)r+j,1), & i = r+p~ (j=i_p,~ p =1, \cdots, s(r))
		\end{cases}
\end{align*}
and 
\begin{align*}
\psi((n+1)r+j,s) = \psi(\tau_r(nr+j,s)) = (T_r\psi)(nr+j,s)= (T_r\psi)_i(n).
\end{align*}
\end{proof}
From Lemma \ref{24/03/04/1:17}, 
we know that $T_r$ is diagonalized by $\mathcal{U} := (\mathcal{F}\mathcal{J})^{-1}$
and $\sigma(T_r) = \{ e^{ik} \mid k\in [-\pi,\pi] \}$. 
These facts imply that $\mathcal{U}$ can decompose $L_G$ into a direct integral over the spectrum $[-\pi,\pi]$ of 
the quasi momentum operator $P$:
\begin{proposition}
\label{24/03/04/1:39}
{\rm
Let $G \in \mathscr{G}$ and let $r$ satisfy \eqref{02/24/23:02}.
Then we have
\begin{equation}
\label{24/03/10/23:25} 
\mathcal{U}^{-1}L_G \mathcal{U} = \int_{[-\pi,\pi]}^{\oplus} L_k dk, 
\end{equation}
where $L_k$ is an operator on $\mathcal{H}_k = \mathbb{C}^{r+s(r)}$ ($k \in [-\pi,\pi]$) and 
\[ \mathcal{U}^{-1}\ell^2(V(G)) = \int_{[-\pi,\pi]}^{\oplus} \mathcal{H}_k dk. \]
}
\end{proposition}
In Subsetion \ref{ss.2.1} 
we provide an explicit formula 
for $L_k$ in Proposition \ref{24/03/04/1:39}.
Let $r$ satisfy \eqref{02/24/23:02} and $U_r(G) = \bigcup_{i=1}^{p(r)} U_i$ be the decomposition in \eqref{24/03/01/15:53}.
Then we use $q(r)$ to denote 
the number of the $U_i$ for which $l_i = \#  U_i$ is odd:
\begin{equation}
\label{24/03/24/21:37} 
q(r) = \# \{ U_i \mid \mbox{$l_i$ is odd} \}. 
\end{equation}
\begin{proposition}
\label{24/03/05/23:40.1}
{\rm
Let  $r$ satisfy \eqref{02/24/23:02}. Then 
\begin{equation} 
\label{24/03/29/15:15}
q(r) = {\rm dim} \ker L_k 
\end{equation}
for any $k \in [-\pi,\pi]$.
}
\end{proposition}
From Proposition \ref{24/03/05/23:40.1}, 
we know that the multiplicity of the zero eigenvalue of $L_k$ is independent of $k \in [-\pi, \pi]$ and equal to $q(r)$.
Thus, we know that there are at most $r+s(r)-q(r)$ 
non-zero eigenvalues of $L_k$.
Let $\lambda_i(k)$ ($i=1,2, \cdots, r+s(r)-q(r)$) be the non-zero eigenvalues of $L_k$
such that 
\[ \lambda_1(k) \geq \lambda_2(k) \geq \cdots \geq \lambda_{r+s(r) - q(r)}(k). \] 
Note that 
\begin{equation} 
\label{24/04/02/09:47}
\sigma(L_G) \setminus \{0\} = \bigcup_{i = 1}^{r+s(r)-q(r)} \{ \lambda_i(k) \mid k \in [-\pi,\pi] \}. 
\end{equation}
Because the number of end vertices equals the number of vertices connected to end vertices,
$q(r)$ is even (resp. odd) if $r+s(r)$ is even (resp. odd).
Hence we know that $r+s(r)-q(r)$ is always even.
Thus, we can define $t(r) \in \mathbb{N}$ as 
\begin{equation} 
\label{24/04/02/23:55}
t(r) = \frac{1}{2}(r+s(r)-q(r)). 
\end{equation}
We are now in a position to state our main result. 
\begin{theorem}
\label{24/03/05/23:40}
{\rm
Let $G \in \mathscr{G}_\times$ and $\lambda_i(k)$ ($i=1,2,\cdots,r+s(r)-q(r)$) be as given above.
Then 
\[ \sigma(L_G) \setminus \{0\} = \bigcup_{i=1}^{t(r)} [a_i, b_i] \cup [-b_i,-a_i], \]
where 
\[ a_i = \begin{cases} 
				\lambda_i(\pi) & \mbox{if $i$ is odd} \\
				\lambda_i(0) & \mbox{if $i$ is even}
			\end{cases}, \quad 
	b_i = \begin{cases} 
				\lambda_i(0) & \mbox{if $i$ is odd} \\
				\lambda_i(\pi) & \mbox{if $i$ is even}
			\end{cases} \]
for  $i = 1,\cdots, t(r)$.
}
\end{theorem}
This theorem has the following corollary:
\begin{corollary}
{\rm 
Let $G \in \mathscr{G}_\times$. Then
$\sigma(L_G)$ has at least one and at most $2t(r)-1 (= r + s(r)-q(r))$
spectral gaps.
}
\end{corollary}
This corollary provides an estimate of the number of spectral gaps.
It remains an open problem to determine the exact number of spectral gaps.

The rest of this paper is organized as follows. 
In Section 2, we present the proof of Proposition \ref{24/03/04/1:39}
with the explicit formula for $L_k$
and derive the characteristic equation for $L_k$.
In Section 3, we prove Proposition \ref{24/03/05/23:40.1},
parts (1) and (2) of Theorem \ref{24/02/29/17:13} and Theorem \ref{24/03/01/16:03}.
In Section 4, we prove Theorem \ref{24/03/05/23:40} and part (3) of Theorem \ref{24/02/29/17:13}.

\section{Characteristic equation}
In subsequent sections, 
unless otherwise specified,
we take it for granted that $G \in \mathscr{G}_\times$.
\subsection{Decomposition of the transition operator} \label{ss.2.1}
In this subsection, we prove Proposition \ref{24/03/04/1:39}.
By definition, we know that
\begin{align*} 
(\mathcal{J}L_G \psi)(n) 
	= \begin{pmatrix} (L_G \psi)_1(n) \\ \mbox{\rotatebox{90}{$\cdots$}} \\ (L_G \psi)_{r+s(r)}(n) \end{pmatrix} 
	= \begin{pmatrix} 
			\frac{1}{d_1} (\psi_r(n-1) + \psi_2(n) + \phi_1(n)) \\ 
			\frac{1}{d_2} (\psi_{1}(n) + \psi_3(n) + \phi_2(n)) \\
				\mbox{\rotatebox{90}{$\cdots$}} \\
			\frac{1}{d_{r-1}} (\psi_{r-2}(n) + \psi_r(n) + \phi_{r-1}(n)) \\
			\frac{1}{d_r} (\psi_{r-1}(n) + \psi_1(n+1) + \phi_r(n)) \\
			\psi_{i_1}(n) \\
				\mbox{\rotatebox{90}{$\cdots$}} \\
			\psi_{i_{s(r)}}(n)
		\end{pmatrix},
\end{align*}
where 
\[ \phi_i(n) = \begin{cases} 
			\psi_{r+p}(n), & i = i_p \\ 
			0, & i \not= i_p
			\end{cases}, \quad p = 1,2, \cdots, s(r). \]
Hence, we have
\begin{align*}
(\mathcal{F}\mathcal{J}L_G \psi)(k)
	= \begin{pmatrix} 
			\frac{1}{d_1} (e^{-ik}\hat{\psi}_r(k) + \hat{\psi}_2(k) + \hat{\phi}_1(k)) \\ 
			\frac{1}{d_2} (\hat{\psi}_1(k) + \hat{\psi}_3(k) + \hat{\phi}_2(k)) \\
				\mbox{\rotatebox{90}{$\cdots$}} \\
			\frac{1}{d_{r-1}} (\hat{\psi}_{r-2}(k) + \hat{\psi}_r(k) + \hat{\phi}_{r-1}(k)) \\
			\frac{1}{d_r} (\hat{\psi}_{r-1}(k) + e^{ik}\hat{\psi}_1(k) + \hat{\phi}_r(k)) \\
			\hat{\psi}_{i_1}(k) \\
				\mbox{\rotatebox{90}{$\cdots$}} \\
			\hat{\psi}_{i_{s(r)}}(k)
		\end{pmatrix} 
		= L_k (\mathcal{F}\mathcal{J}\psi)(k),
\end{align*}
where
\begin{equation}
\label{24/03/29/11:38.0}
 L_k = \begin{pmatrix} L_{11} & L_{12} \\ L_{21} & L_{22} \end{pmatrix} 
\end{equation}
with
\begin{align}
\label{24/03/29/11:38.1}
L_{11}
	= \begin{pmatrix} 
		0 & \frac{1}{d_1} & 0 &  \cdots & 0 & \frac{e^{-ik}}{d_1} \\ 
		\frac{1}{d_2} & 0 & \mbox{\rotatebox{140}{$\cdot\cdots$}} &  0 & \cdots & 0 \\
		0 & \mbox{\rotatebox{140}{$\cdot\cdots$}} & \mbox{\rotatebox{140}{$\cdot\cdots$}} & \mbox{\rotatebox{140}{$\cdot\cdots$}} 
			& \mbox{\rotatebox{140}{$\cdot\cdots$}} & \mbox{\rotatebox{90}{$\cdots$}} \\
		\mbox{\rotatebox{90}{$\cdots$}} & \mbox{\rotatebox{140}{$\cdot\cdots$}} & \mbox{\rotatebox{140}{$\cdot\cdots$}} 
			& \mbox{\rotatebox{140}{$\cdot\cdots$}} & \mbox{\rotatebox{140}{$\cdot\cdots$}}  & 0\\
		0 & \cdots & 0 & \mbox{\rotatebox{140}{$\cdot\cdots$}} & \mbox{\rotatebox{140}{$\cdot\cdots$}} & \frac{1}{d_{r-1}} \\
		\frac{e^{ik}}{d_r} & 0 & \cdots & 0 & \frac{1}{d_r} & 0
		\end{pmatrix}, 
\\
\label{24/03/29/11:38.2}
L_{22}
	= \begin{pmatrix}
		0 & \frac{1}{d_{r+1}} & 0 &  \cdots & 0 \\ 
		\frac{1}{d_{r+2}} & 0 & \mbox{\rotatebox{140}{$\cdot\cdots$}} & \mbox{\rotatebox{140}{$\cdot\cdots$}} & \mbox{\rotatebox{90}{$\cdots$}}  \\
		0 & \mbox{\rotatebox{140}{$\cdot\cdots$}} & \mbox{\rotatebox{140}{$\cdot\cdots$}} & \mbox{\rotatebox{140}{$\cdot\cdots$}} 
			& 0 \\
		\mbox{\rotatebox{90}{$\cdots$}} & \mbox{\rotatebox{140}{$\cdot\cdots$}} & \mbox{\rotatebox{140}{$\cdot\cdots$}} 
			& \mbox{\rotatebox{140}{$\cdot\cdots$}} & \frac{1}{d_{r+s(r)-1}} \\
		0 & \cdots & 0 & \frac{1}{d_{r+s(r)}} &  0 \\
		\end{pmatrix}.
\end{align}
Here $L_{12} = (b_{i.j})$ and $L_{21}=(c_{i,j})$ are given by
\begin{align}
\label{24/03/29/11:38.3}
& b_{i,j} = \begin{cases} \frac{1}{d_i}, & i = i_p, ~ j = p, \quad p = 1, \cdots, s(r) \\ 0, & \mbox{otherwise}, \end{cases} \\
\label{24/03/29/11:38.4}
& c_{i,j} = \begin{cases} \frac{1}{d_{r+i}}, & i = p, ~ j = i_p, \quad p = 1, \cdots, s(r) \\ 0, & \mbox{otherwise}. \end{cases} 
\end{align}
Thus, we know that $\mathcal{U}^{-1}L_G \mathcal{U} \cdot (\mathcal{U}^{-1}\psi)(k) = L_k  (\mathcal{U}^{-1}\psi)(k)$,
which establishes Proposition \ref{24/03/04/1:39}.

We conclude  this subsection with the following lemma:
\begin{lemma} \label{24.07/09/10:53}
{\rm
Let $k=0$. Then $L_0$ is a stochastic matrix,
{\it i.e.} the elements of $L_0$ are all nonnegative and the sum of each row (resp. column) of $L_0$ is $1$.
}
\end{lemma}

\subsection{Characteristic equation}

Let 
\[ \mathcal{U}^{-1} L_G \mathcal{U} = \displaystyle \int^\oplus_{[-\pi,\pi]} L_k dk \]
be the decomposition in \eqref{24/03/10/23:25}.
In this subsection, 
we investigate the characteristic polynomial for $L_k$ 
and prove the following:
\begin{proposition}
\label{24/03/19.17:14}
{\rm
Let $G \in \mathscr{G}$.
Then there exist a positive constant $A_{(r+s(r)-\star)/2}$ 
and  nonnegative constants $A_m$ ($m=0,1,2, \dots, (r+s(r)-\star)/2-1$)
independent of $k \in [-\pi,\pi]$ such that
$\lambda = \lambda(k)$ is an eigenvalue of $L_k$
if and only if 
\begin{equation}
\label{24/03/22/20:51} 
\sum_{m=0}^{ (r+s(r)-\star)/2} A_m \lambda^{2m+\star} - 2 (\cos k ) \lambda^{s(r)} = 0, 
\end{equation}
where 
\begin{equation} 
\label{24/04/22/12:03}
\star = \begin{cases} 
		0, & \mbox{if $r+s(r)$ is even}, \\
		1, & \mbox{if $r+s(r)$ is odd}.
		\end{cases}
\end{equation}
}
\end{proposition}
In this proposition, 
$r+s(r)-\star$ is always even; hence, $(r+s(r)-\star)/2 \in \mathbb{N}$.
To prove this proposition, we introduce a directed weighted graph $G_k = (V(G_k), A(G_k))$ for $L_k$ as follows.
We use the ordered pair $[x,y]$ of two vertices $x$ and $y$
to denote the directed edge whose initial and terminal vertices are $x$ and $y$, respectively.
The set of vertices $V(G_k)$ of $G_k$ is $V(G_k) = \{ 1,2,\cdots,r+s(r) \}$
and the set of directed edges of $G_k$ is $A(G_k) = \{ [i,i^\prime] \mid i \sim i^\prime,~ i,i^\prime \in V(G_k) \}$,
where $i \sim i^\prime$ if $i$ and $i^\prime$ satisfy one of the following three conditions:
\begin{itemize}
\item[(a)] There exists an $n$ such that $\iota^{-1}(n,i) \sim \iota^{-1}(n,i^\prime)$ as vertices of $G$.
\item[(b)] $(i,i^\prime) = (1,r)$ or $(i,i^\prime) = (r,1)$.
\item[(c)] $i = i^\prime$, $i = 1,2, \cdots, r+s(r)$. 
\end{itemize}
Note that in the case of (c), the edge $[i,i]$ is a loop.
We define the weight $a_{i,i^\prime}$ of the edge $[i,i^\prime]$ as
\[ a_{i,i^\prime} 
	= \begin{cases} 
		-d_i \lambda & \mbox{if $i = i^\prime$,} \\ 
		1 & \mbox{if $\iota^{-1}(n,i) \sim \iota^{-1}(n,i^\prime)$,} \\ 
		e^{ik} & \mbox{if $(i,i^\prime) = (r,1)$,} \\ 
		e^{-ik} & \mbox{if $(i,i^\prime) = (1,r)$}  \end{cases} \]
with $\lambda \in \mathbb{C}$.
As a notational convention, we set $a_{i,i^\prime} = 0$ if $[i,i^\prime] \not\in A(G_k)$.
We define square matrices of order $r+s(r)$ by
\[ A_k(\lambda) = A_k - \lambda D, \quad \lambda \in \mathbb{C},\]
where $A_k = (a_{i,i^\prime})$ and $D = {\rm diag}(d_1, \cdots, d_{r+s(r)})$.
We observe that from direct calculation,
\[ A_k(\lambda) = (a_{i,i^\prime}).
\]

We now establish
\begin{lemma}
\label{24/03/22/16:38}
{\rm
$\lambda$ is an eigenvalue of $L_k$ if and only if $A_k(\lambda)$ has a zero eigenvalue.
}
\end{lemma}
\begin{proof}
Since $D L_k = A_k$, we have
\begin{equation}
\label{24/04/04/16:59} 
A_k(\lambda) = D(L_k - \lambda). 
\end{equation}
$\lambda$ is an eigenvalue of $L_k$ if and only if $ {\rm det} (L_k - \lambda) = 0$. 
Because $D$ is invertible and ${\rm det}(L_k - \lambda) = {\rm det} (D^{-1}) \cdot {\rm det}( A_k(\lambda))$,
we have the desired result.
\end{proof}
From Lemma \ref{24/03/22/16:38}, we know that the solution of the equation 
\[ {\rm det}(A_k(\lambda)) = 0 \]
is an eigenvalue of $L_k$. 
Note that
\begin{equation}
\label{24/03/33/16:35} 
{\rm det}(A_k(\lambda)) 
	= \sum_{\sigma \in \mathfrak{S}_{r+s(r)}} A_\sigma, 
\end{equation}
where  $\mathfrak{S}_n$ denotes the set of all permutations of degree $n$ and 
\[ A_\sigma = {\rm sgn} (\sigma) \prod_{i=1}^{r+s(r)} a_{i,\sigma(i)},
	\quad \sigma \in \mathfrak{S}_{r+s(r)}.\]

For $\sigma \in \mathfrak{S}_{r+s(r)}$,
there exist integers $1 \leq n \leq  r+s(r)$ and $n_j \geq 1$ ($j=1, \cdots, n$) such that
$r+s(r) = \sum_{j = 1}^n \mu_j$ and
\begin{equation}
\label{24/03/19/14:39}
\sigma = (i_1^{(1)}, \cdots, i_{\mu_1}^{(1)}) \cdots (i_1^{(n)}, \cdots, i_{\mu_n}^{(n)}),
\end{equation}
where, for $\mu_j \geq 2$, $(i_1^{(j)}, \cdots, i_{\mu_j}^{(j)})$ denotes
a cycle such that $i_1^{(j)}, \cdots, i_{\mu_j}^{(j)}$ satisfy
\[ \sigma(i_1^{(j)}) = i_2^{(j)}, \quad \sigma(i_2^{(j)}) = i_3^{(j)}, \cdots, \sigma(i_{\mu_j}^{(j)}) = i_1^{(j)}, \] 
and for $\mu_j=1$, $\sigma(i_{1}^{(j)}) = i_{1}^{(j)}$.
For $\sigma \in \mathfrak{S}_{r+s(r)}$, we define $n$ disjoint graphs $G_k^{(j)}$ ($j=1,2, \cdots,n$) by
\begin{equation}
\label{24/03/19/15:32} 
V(G_k^{(j)}) = \{ i_1^{(j)}, \cdots, i_{\mu_j}^{(j)} \}, 
	\quad A(G_k^{(j)}) = \{ [i_{\nu}, \sigma(i_{\nu})] \mid \nu = 1, \cdots, \mu_j \}. 
\end{equation}
We say that 
a finite directed graph $H=(V(H),E(H))$ with vertices $V(H) = \{h_1, \cdots, h_m \}$ 
is a cycle of length $m$
if  
\[ A(H) = \{ [h_1, h_2], [h_2, h_3], \cdots, [h_m,h_1] \}. \]
In particular, we call a cycle of length $1$ a loop.
We use $H_1 \dot{\cup} H_2$ to denote the disjoint union of two disjoint graphs $H_1$ and $H_2$,
{\rm i.e.} a graph whose vertices  and edges are $V(H_1 \dot{\cup} H_2) = V(H_1) \cup V(H_2)$
and $A(H_1 \dot{\cup} H_2) = A(H_1) \cup A(H_2)$, respectively.
We say that $H$ is a spanning subgraph of $G$ if $V(H) = V(G)$ and $A(H) \subset A(G)$.
From definition, we know that
$G_k^{(i)}$ ($j=1,\cdots,n$) are cycles, and
if $\mu_j = 1$, then $G_k^{(j)}$ is a loop.
Because $r+s(r) = \sum_{j = 1}^n \mu_j$ and the graphs $G_k^{(j)}$ ($j=1,2, \cdots,n$) are disjoint,
we have $V(G_k) = \cup_{j=1}^n V(G_k^{(j)})$.
Thus, we know that each $\sigma \in \mathfrak{S}_{r+s(r)}$
corresponds to a disjoint union $G_\sigma = G_k^{(1)} \dot{\cup} \cdots \dot{\cup} G_k^{(n)}$
with $G_k^{(j)}$ defined by \eqref{24/03/19/15:32}.
Let $\mathscr{G}_k$ be the set of graphs $\tilde{G}$ that satisfy $V(\tilde{G}) = V(G_k)$. 
Then we define a map $Q:\mathfrak{S}_{r+s(r)} \rightarrow \mathscr{G}_k$ as
\[ Q(\sigma) = G_\sigma, \quad \sigma \in \mathfrak{S}_{r+s(r)}. \] 
From its construction, the map $Q$ is injective.
Conversely, we suppose that disjoint graphs $G_k^{(j)}$ ($j=1,2, \cdots,n$) are cycles 
with vertices $\{ i_1^{(j)}, \cdots, i_{\mu_j}^{(j)} \} \subset V(G_k)$
such that $r+s(r) = \sum_{j = 1}^n \mu_j$.
Then the disjoint union $G_k^{(1)} \dot{\cup} \cdots \dot{\cup} G_k^{(n)}$
satisfies $V(G_k) = \cup_{j=1}^n V(G_k^{(j)})$.
Let $\sigma \in \mathfrak{S}_{r+s(r)}$ be a permutation defined by \eqref{24/03/19/14:39}.
Then we have $Q(\sigma) = G_\sigma = G_k^{(1)} \dot{\cup} \cdots \dot{\cup} G_k^{(n)}$.
Hence, we know that $Q$ is bijective.
The following lemma also holds:

\begin{lemma}
\label{2013/04/15/12:37}
{\rm 
Let $\sigma \in \mathfrak{S}_{r+s(r)}$, $A_\sigma$ and 
$G_\sigma = G_k^{(1)} \dot{\cup} \cdots \dot{\cup} G_k^{(n)}$ be as above.
Then $A_\sigma \not=0$ if and only if 
$G_\sigma$ is a spanning subgraph of $G_k$.
}
\end{lemma}
\begin{proof}
It is clear that $A_\sigma \not= 0$ if and only if $a_{i,\sigma(i)} \not= 0$ for any $i \in \{ 1,2, \cdots, r+s(r) \}$.
Because, by definition, $a_{i,\sigma(i)} \not= 0$ if and only if $[i,\sigma(i)] \in A(G_k)$,
we have the desired result.
\end{proof}

We can now present the 
\begin{proof}[Proof of Proposition \ref{24/03/19.17:14}]
{\rm
By virtue of Lemma \ref{2013/04/15/12:37}, it suffices to consider the case where 
$G_\sigma = G_k^{(1)} \dot{\cup} \cdots \dot{\cup} G_k^{(n)}$ is a spanning subgraph of $G_k$.

Let $H_k$ be a cycle given by
\begin{align*} 
V(H_k) =  \{ 1,2, \cdots,r \},
	\quad A(H_k) = \{ [1,2], [2,3], \cdots, [r,1] \}
\end{align*}
We first consider a spanning subgraph $G_k^{(1)} \dot{\cup} \cdots \dot{\cup} G_k^{(n)}$ of $G_k$
that includes a cycle $H_k$.
Without loss of generality, we can set $H_k = G_k^{(1)}$.
In this case, $G_k^{(j)}$ ($j\geq 2$) cannot be cycles of length greater than $2$ 
because $V(G_k^{(j)}) \subset \{r+1, \cdots, r+s(r)\}$.
Hence, we observe that  $n = s(r)+1$ and $G_k^{(J+1)}$ ($j =1, \cdots, s(r)$) 
are loops with $V(G_k^{(J+1)}) = \{r+j\}$.
Let $\sigma_+ \in \mathfrak{S}_{r+s(r)}$ be 
the permutation such that $G_{\sigma_+} = G_k^{(1)} \dot{\cup} \cdots \dot{\cup} G_k^{(s(r)+1)}$
with $G_k^{(j)}$ ($j=1, \cdots, s(r) +1$) as mentioned above. 
Then we have $\sigma_+(1)=2$, $\sigma_+(2) = 3$, $\cdots$, $\sigma_+(r)=1$ and
\[ \prod_{i=1}^r a_{i,\sigma_+(i)} = e^{ik}. \]  
Because $\sigma_+(i)=i$ holds for $r+1 \leq i \leq r+s(r)$, 
we obtain 
\[ \prod_{i=r+1}^{r+s(r)} a_{i,\sigma_+(i)} = (-\lambda)^{s(r)}. \] 
Note that ${\rm sgn}\sigma_+ = (-1)^{r-1}$ since the length of the cycle $H_k$ is $r$.
Thus, we have
\[ A_{\sigma_+} = (-1)^{r-1}(-\lambda)^{s(r)} e^{ik}. \]
Similarly, we set $\sigma_- \in \mathfrak{S}_{r+s(r)}$ such that
$G_{\sigma_-} = G_k^{(1)} \dot{\cup} \cdots \dot{\cup} G_k^{(s(r)+1)}$
with $G_k^{(1)}$ replaced by $H_k^\prime$, 
where $H_k^\prime$ is a cycle given by
\begin{align*} 
 V(H_k^\prime) = \{ 1,2, \cdots,r \},
 	\quad A(H_k^\prime) = \{ [1,r], [r,r-1], \cdots, [2,1] \}. 
\end{align*}
Then we have
\[ A_{\sigma_-} = (-1)^{r-1}(-\lambda)^{s(r)} e^{-ik}. \]
Hence, we also have
\begin{equation}
\label{24/03/22/16:21} 
A_{\sigma_+} + A_{\sigma_-} = 2(-1)^{r+s(r)-1}(\cos k) \lambda^{s(r)}. 
\end{equation}
Note that $H_k$ (resp. $H_k^\prime$) is the only cycle of length greater than 3 
that contains $[1,r]$ (resp. $[r,1]$);
a cycle $H_0$ of length 2 that contains $[1,r]$ or $[r,1]$
is given by 
\[ V(H_0) = \{1,r\}, \quad A(H) = \{[1,r], [r,1] \}. \]
Hence, we know that if $A_\sigma$ depends on $k \in [-\pi,\pi]$,
then the edges $A(G_\sigma)$ of the spanning subgraph $G_\sigma$ includes either $H_k$ or $H_k^\prime$.
Let
\[ \mathfrak{S}_{r+s(r)}^\prime = \{ \sigma \in \mathfrak{S}_{r+s(r)} \mid \mbox{$\sigma(1) =r$, $\sigma(r)=1$} \}. \]
From the argument presented above, we have
\[
 \{ \sigma \in \mathfrak{S}_{r+s(r)} \mid \mbox{$\sigma(1) =r$ or $\sigma(r)=1$} \}
	= \mathfrak{S}_{r+s(r)}^\prime \cup \{ \sigma_+, \sigma_- \}. 
\]

Next we consider a cycle whose edges include both $[1,r]$ and $[r,1]$.
The cycle $H_0$ defined above is the only cycle that includes both $[1,r]$ and $[r,1]$.
We set $G_k^{(1)} = H_0$. Let $G_k^{(j)}$ ($j=2,\cdots,n$) be disjoint cycles and
assume that $\tilde{G} := H_{(0)} \dot{\cup} G_k^{(2)} \dot{\cup} \cdots \dot{\cup} G_k^{(n)}$
is a spanning subgraph of $G_k$. 
Because there is no cycle of length grater than 3 that does not include edges $[1,r]$ and $[r,1]$, 
the length of each $G_k^{(j)}$ ($j=2,\cdots,n$) is at most 2.
We now claim that the number of loops included in the spanning subgraph $\tilde{G}$ is even if $r+s(r)$ is even.
Note that if $r +s(r)$ is even, 
then $\#  V(G_k) \setminus V(H_0) = r+s(r)-2$, which is even. 
If the spanning subgraph $\tilde{G}$ has $2\nu+1$ loops, 
then cycles of length at most 2 must be constructed from the remaining $r+s(r)-2-(2m+1)$ vertices.
Because $r+s(r)-2-(2m+1)$ is odd, 
these cycles include an odd number of loops.
Thus, the claim is proven.
Similarly, we can prove that the number of loops included in the spanning subgraph $\tilde{G}$ is odd if $r+s(r)$ is odd.
Let $2m+\star$ ($0 \leq m \leq (r+s(r)-2-\star)/2$) be the number of loops included in the spanning subgraph $\tilde{G}$,
where $\star$ is defined as given in \eqref{24/04/22/12:03}.
Then the number of cycles of length $2$ included in the spanning subgraph $\tilde{G}$  is $(r+s(r)-(2m+\star))/2$.
Let $\sigma_{\tilde{G}}$ be the permutation which corresponds to $\tilde{G}$.
Then $\sigma_{\tilde{G}} \in \mathfrak{S}_{r+s(r)}$ and
we have ${\rm sgn}\sigma_{\tilde{G}} = (-1)^{1+(r+s(r)-(2m+\star))/2}$.
Because $\sigma_{\tilde{G}}(1) = r$ and $\sigma_{\tilde{G}}(r) = 1$,
it follows that $a_{1,\sigma_{\tilde{G}(1)}}a_{r,\sigma_{\tilde{G}(r)}} = e^{-ik} \cdot e^{ik} = 1$.
Hence, we can write
\begin{align} 
\sum_{\sigma \in \mathfrak{S}_{r+s(r)}^\prime}A_\sigma 
& = \sum_{m=0}^{(r+s(r)-2-\star)/2} B_m (-\lambda)^{2m+\star}(-1)^{(r+s(r)-(2m+\star))/2} \notag \\
\label{24/03/22/16:25} 
& = \sum_{m=0}^{(r+s(r)-2-\star)/2} B_m (-1)^{(r+s(r)+\star-2m)/2} \lambda^{2m+\star} 
\end{align}
with constants $B_m \geq 0$ that are independent of $k \in [-\pi,\pi]$.

Let us calculate $A_\sigma$ for $\sigma \in \mathfrak{S}_{r+s(r)}$ 
such that $\sigma \not\in \mathfrak{S}_{r+s(r)}^\prime$
and $\sigma \not=\sigma_\pm$.
In this case, we have $\sigma(1) \not=r$ and $\sigma(r)\not=1$.
By reasoning similar to that mentioned above, 
we can assume that there are $2m + \star$ loops 
and $(r+s(r) - (2m+\star))/2$ cycles of length $2$ in the spanning subgraph that corresponds to $\sigma$,
where $0 \leq m \leq (r+s(r)-\star)/2$.
Hence, we obtain
\begin{align}
\sum_{\sigma: \sigma(1)\not=r, \sigma(r)\not=1}A_\sigma 
& = \sum_{m=0}^{(r+s(r)-\star)/2} C_m (-\lambda)^{2m+\star}(-1)^{(r+s(r)-(2m+\star))/2} \notag \\
\label{24/03/22/16:26} 
& = \sum_{m=0}^{(r+s(r)-\star)/2} C_m (-1)^{(r+s(r)+\star-2m)/2} \lambda^{2m+\star}
\end{align}
with $C_m \geq 0$ independent of $k \in [-\pi,\pi]$.
We observe that 
\begin{equation}
\label{24/03/22/21:41} 
C_{(r+s(r)-\star)/2} = \prod_{i=1}^{r+s(r)} d_i > 0 
\end{equation}
because the cycles included in the corresponding spanning subgraph are loops.  

From \eqref{24/03/22/16:21}, \eqref{24/03/22/16:25} and \eqref{24/03/22/16:26},
we have
\begin{align*} 
\sum_{\sigma \in \mathfrak{S}_{r+s(r)}}A_\sigma 
	& =  C_{(r+s(r)-\star)/2} (-1)^\star \lambda^{r+s(r)} \\
		& \quad + \sum_{m=0}^{(r+s(r)-2-\star)/2} (B_m + C_m)(-1)^{(r+s(r)+\star-2m)/2} \lambda^{2m+\star} \\
			& \quad + 2(-1)^{r+s(r)-1}(\cos k) \lambda^{s(r)}
\end{align*}
Based on Lemma \ref{24/03/22/16:38} and \eqref{24/03/33/16:35}, 
we know that $\lambda$ is an eigenvalue of $L_k$
if and only if 
\begin{align*} 
& C_{(r+s(r)-\star)/2} \lambda^{r+s(r)} \\
		& \quad + \sum_{m=0}^{(r+s(r)-2-\star)/2} (-1)^{-\star + (r+s(r)+\star-2m)/2} (B_m + C_m) \lambda^{2m+\star} \\
			& \quad  - 2 (-1)^{-\star + r+s(r)}(\cos k) \lambda^{s(r)} = 0.
\end{align*}
Let $A_{(r+s(r)-\star)/2} = C_{(r+s(r)-\star)/2}$ and 
\begin{equation}
\label{24/04/03/14:11}
A_m =  (-1)^{(r+s(r)-2m-\star)/2}(B_m + C_m), \quad 0 \leq m \leq (r+s(r)-2-\star)/2.
\end{equation}
From \eqref{24/03/22/21:41}, we know that $A_{(r+s(r)-\star)/2}$ is strictly positive.
Since $-\star + r + s(r)$ is always even and $(-1)^{-\star + r+s(r)} = 1$,
Proposition \ref{24/03/19.17:14} is proven.
}
\end{proof}

\subsection{Lowest degree of the characteristic polynomial}
In this subsection, we prove the following proposition:
\begin{proposition}
\label{24/03/27/10:36}
{\rm
Let $A_m$ be as in Proposition \ref{24/03/19.17:14}
and let $q(r)$ be defined by \eqref{24/03/24/21:37}.
Then the following holds:  
\[ 
q(r) = \inf\{ 2m+\star \mid A_m \not= 0 \}.
\] 
}
\end{proposition}

To prove this proposition, let $U_r(G) = \bigcup_{i=1}^{p(r)} U_i$, where
\[ U_i = \left\{ \left(m_{1+ \sum_{j=1}^{i-1} l_j},0 \right), \left(m_{2+ \sum_{j=1}^{i-1} l_j},0 \right),
	\cdots, \left(m_{\sum_{j=1}^{i} l_j},0 \right) \right\} \]
is the decomposition defined in \eqref{24/03/01/15:53} and \eqref{24/03/24/15:56}.
Then we define a corresponding decomposition for $G_k$ by
\[ U_k = \bigcup_{i=1}^{p(r)} U_k^{(i)}, \]
where
\[ U_k^{(i)} = \left\{ m_{1+ \sum_{j=1}^{i-1} l_j}, m_{2+ \sum_{j=1}^{i-1} l_j},
	\cdots, m_{\sum_{j=1}^{i} l_j} \right\}. \]
For each $U_k^{(i)}$, we define cycles $K^{(i)}_\alpha$ 
whose vertices are included in $U_k^{(i)}$ as follows: 
If $\#  U_k^{(i)} = \#  U_i$ is even, then we set
\[ V(K^{(i)}_\alpha) = \{ m_{2\alpha-1 + \sum_{j=1}^{i-1} l_j}, m_{2\alpha+ \sum_{j=1}^{i-1} l_j} \},
	\quad  
	1 \leq \alpha \leq m_{\sum_{j=1}^{i} l_j}/2. \]
If $\#  U_k^{(i)} = \#  U_i$ is odd, then we define a loop $K^{(i)}_1$ 
with $V(K^{(i)}_1) = \{ m_{1+ \sum_{j=1}^{i-1} l_j} \}$
and  cycles $K^{(i)}_\alpha$ of length 2 
with
\[ V(K^{(i)}_\alpha) = \{ m_{2\alpha + \sum_{j=1}^{i-1} l_j}, m_{1 + 2\alpha+ \sum_{j=1}^{i-1} l_j} \},
	\quad  
	1 \leq \alpha \leq (m_{\sum_{j=1}^{i} l_j}-1)/2. \]
Let $1 = i_1 < i_2 < \cdots < i_{s(r)}$ be as in \eqref{24/03/27/16:20}.
Then, by definition, we know that $i_j \not\in U_k$ and 
\[ [i_j, r+j], [r+j,i_j] \in A(G_k), \quad j=1, \cdots,s(r). \]
We define cycles $H_k^{(j)}$ ($j=1, \cdots,s(r)$) of length $2$ with
\[ V(H_k^{(j)}) = \{ i_j, r+j \}, \quad A(H_k^{(j)}) = \{[i_j, r+j], [r+j,i_j] \}. \]
Since $V(G_k) \setminus ( \cup_j V(H_k^{(j)})) = U_k$, we have a spanning graph
\begin{equation} 
\label{24/03/28/16:13}
K = \left( \dot{\bigcup_j} H_k^{(j)} \right) \dot{\cup} \left( \dot{\bigcup_{i,\alpha}} K_\alpha^{(i)} \right). 
\end{equation}
By definition, $K$ has exactly $q(r)$ loops. 
In general, $K$ is not a unique spanning subgraph of $G_k$ that has exactly $q(r)$ loops. 
Indeed,  if $q(r) \geq 1$, then we can construct such a spanning subgraph of $G_k$ as follows:
Since $q(r) \geq 1$, there is a number $i_0$ such that $\#  U_k^{(i)}$ is odd. 
Hence, by definition, $ K_1^{(i_0)}$ is a loop with the vertex $m_{1+ \sum_{j=1}^{i-1} l_j}$.
Then there is a cycle $H_k^{(j_0)}$ with vertices $i_{j_0}, r+j_0$ such that 
\[ [i_{j_0}, m_{1+ \sum_{j=1}^{i-1} l_j} ] \in A(G_k). \]
Hence, we can define a spanning subgraph $K^\prime$ with exactly $q(r)$ loops by
\[ K^\prime
= \left( \dot{\bigcup_{j \not= j_0} } H_k^{(j)} \right) \dot{\cup} \left( \dot{\bigcup_{(i,\alpha) \not= (i_0,1)}} K_\alpha^{(i)} \right)
	\dot{\cup} M, \]
where $M$ is the disjoint union of a loop with vertex $r+j_0$ 
and a cycle with vertices $i_{j_0}$ and $m_{1+ \sum_{j=1}^{i-1} l_j}$.

Let 
\[ m_0 = \frac{1}{2}(q(r) - \star). \]
Since $q(r) -\star$ is always even, we know that $m_0$ is a nonnegative integer.
In particular, $m_0 = 0$ if $q(r)=0$ or $q(r)= 1$.
We shall now prove
\begin{lemma}
\label{24/03/29/09:59}
{\rm
The following hold:
\begin{itemize}
\item[(a)] If $q(r) = 0$, then $|A_0| = 1$.
\item[(b)] If $q(r)\geq 1$, then  $|A_{m_0}| > 2$.
\end{itemize}
In particular, we have 
\begin{equation}
\label{24/03/28/10:40.1}  
\inf\{ 2m+\star \mid A_m \not= 0 \} \leq q(r). 
\end{equation}
}
\end{lemma}
\begin{proof}
Let $a_{\sigma}$ denote the absolute value of the coefficient of $A_\sigma$.
Then for any spanning subgraph $\tilde{K}$ of $G_k$  which has exactly $q(r)$ loops,
we have $A_{\sigma_{\tilde{K}}} = ({\rm sgn}\sigma_{\tilde{K}}) a_{\sigma_{\tilde{K}}} \lambda^{q(r)}$.
We assume that $\sigma_{\tilde{K}} \not= \sigma_\pm$ if $q(r) = s(r)$.
Because every $\tilde{K}$ has the same number of cycles of length $2$ as $K$, 
we know that $\sigma_{\tilde{K}} =  \sigma_{K}$.
Hence we have $|A_{m_0}| = \sum_{\tilde{K}} a_{\sigma_{\tilde{K}}}$.

To prove (a), we assume that $q(r) = 0$. 
Then $\#  U_k^{(i)}$ is even for $i=1,\cdots,p(r)$
and all disjoint subgraphs included in $K$ are cycles of length $2$.
Since all the spanning subgraphs of $G_k$ that do not have loops include $H_k^{(j)}$ ($j=1,\cdots,s(r)$), 
$K$ is a unique spanning subgraph of $G_k$ without loops.
Hence $|A_0| = a_{\sigma_K} = 1$. 

To prove (b), we assume that $q(r)\geq 1$. 
Then, as mentioned above, $K$ and $K^\prime$ have exactly $q(r)$ loops.
We know $a_{\sigma_K} \geq 2$, 
because $K_1^{(i_0)}$ is the loop with the vertex $m_{1+ \sum_{j=1}^{i-1} l_j}$
and ${\rm deg}~ m_{1+ \sum_{j=1}^{i-1} l_j} = 2$.
Hence, $|A_{m_0}| \geq a_{\sigma_K} + a_{\sigma_{K^\prime}} > 2$.

Because $q(r) = 2m_0+\star$,
\eqref{24/03/28/10:40.1} is obtained from (a) and (b).
\end{proof}

We use $ \mathcal{K}(\{M^{(i)} \}_{i=1}^s) $ to denote the set of all spanning subgraphs of $G_k$ that include 
disjoint graphs $M^{(i)}$ ($i=1, \cdots,s$).
Let $q(\tilde{G})$ be the number of loops included in a spanning subgraph $\tilde{G}$ of $G_k$.
For any spanning subgraph $\tilde{G}$ in $\mathcal{K}(\{H_k^{(j)}\}_{j=1}^{s(r)})$,
the lengths of the cycles included in $\tilde{G}$ are all less than or equal to $2$
because $[1,r+1] \in A(H_k^{(1)})$.
Hence, $q(\tilde{G})$ equals the number of loops with a vertex in $U_k$.
If $l_i = \#  U_k^{(i)}$ is odd, 
then $\tilde G \in \mathcal{K}(\{H_k^{(j)}\}_{j=1}^{s(r)})$ has at least one loop with a vertex of $U_k^{(i)}$;
if $l_i = \#  U_k^{(i)}$ is odd, 
then $\tilde G$ contains no loop with a vertex of $U_k^{(i)}$.
Therefore, $q(r)$ is the minimum of the number of loops included in a spanning subgraph in $\mathcal{K}(\{H_k^{(j)}\}_{j=1}^{s(r)})$:
\begin{equation} 
\label{24/03/38/14:29}
q(r) = \min \{ q(\tilde{G}) \mid \tilde{G} \in \mathcal{K}(\{H_k^{(j)}\}_{j=1}^{s(r)}) \}.
\end{equation}
Let $1 \leq j_1 < j_2 < \cdots < j_s \leq s(r)$ ($1 \leq s \leq s(r)$)
and $i_0 = 1,\cdots,s$.
Because $\mathcal{K}(\{H_k^{(j_i)}\}_{i = 1}^s) \subset \mathcal{K}(\{H_k^{(j_i)}\}_{i \not= i_0})$,
we have 
\begin{equation}
\label{24/03/28/17:05} 
\min \{  q(\tilde{G}) \mid 
		\tilde{G} \in \mathcal{K}(\{H_k^{(j_i)}\}_{i \not= i_0}) \}
	\leq \min \{  q(\tilde{G}) \mid \tilde{G} \in \mathcal{K}(\{H_k^{(j_i)}\}_{i = 1}^s) \}. 
\end{equation}
Indeed, the converse inequality of \eqref{24/03/28/17:05} holds as well as the following:
\begin{lemma}
\label{24/03/28/14:03}
{\rm
Let $1 \leq j_1 < j_2 < \cdots < j_s \leq s(r)$ ($1 \leq s \leq s(r)$).
Then 
\begin{align*} 
\min_{i_0=1,2,\cdots,s} 
	\min \{  q(\tilde{G}) \mid 
		\tilde{G} \in \mathcal{K}(\{H_k^{(j_i)}\}_{i \not= i_0}) \} 
	= \min \{  q(\tilde{G}) \mid \tilde{G} \in \mathcal{K}(\{H_k^{(j_i)}\}_{i = 1}^s) \}. 
\end{align*}
}
\end{lemma}
Before proving Lemma \ref{24/03/28/14:03}, we complete the proof of Proposition \ref{24/03/27/10:36}:
\begin{proof}
{\rm
We only require to prove the converse of the inequality \eqref{24/03/28/10:40.1}.
From Lemma \ref{24/03/28/14:03} and \eqref{24/03/38/14:29}, 
we know that
\[ \min_{s=1, \cdots, s(r)}~\min_{1 \leq j_1 < \cdots < j_s \leq s(r)} 
		\min \{  q(\tilde{G}) \mid \tilde{G} \in \mathcal{K}(\{H_k^{(j_i)}\}_{i = 1}^s) \} = q(r). \]
Note that the set of all spanning subgraphs of $G_k$ is
\[ \{ G_{\sigma_+}, G_{\sigma_-} \} \cup 
		\left[ \bigcup_{s=1}^{s(r)}~\bigcup_{1 \leq j_1 < \cdots < j_s \leq s(r)}
			\mathcal{K}(\{H_k^{(j_i)} \}_{i=1}^s) \right]. \]
Since the number of loops included in a spanning subgraph $\tilde{G}$ equals
the degree of $A_{Q^{-1}(\tilde{G})}$ in $\lambda$,
we have
\[ \inf\{ 2m+\star \mid A_m \not= 0 \} \geq q(r). \]
This completes the proof of Proposition \ref{24/03/27/10:36}.
}
\end{proof}
\begin{proof}[Proof of Lemma \ref{24/03/28/14:03}.]
We prove the lemma for the case where $s=s(r)$;
the other case can be similarly proven.
It suffices to prove that
\begin{equation}
\label{24/03/28/15:03} 
\min \{  q(\tilde{G}) \mid \tilde{G} \in \mathcal{K}(\{H_k^{(j)}\}_{j \not= j_0}) \} \geq q(r). 
\end{equation}
If $\tilde{G} \in \mathcal{K}(\{H_k^{(j)}\}_{j \not= j_0}) \setminus \mathcal{K}(\{H_k^{(j)}\}_{j =1}^{s(r)})$,
then $\tilde{G}$ includes the loop  with the vertex $r+j_0$,
because $\tilde{G}$ does not include $H_k^{(j_0)}$.
Hence, we have
\[ \mathcal{K}(\{H_k^{(j)}\}_{j \not= j_0}) 
	= \mathcal{K}(\{H_k^{(j)}\}_{j =1}^{s(r)}) \cup \mathcal{K}(\{\tilde{H}_k^{(j)}\}_{j =1}^{s(r)}), \]
where $\tilde{H}_k^{(j)} = H_k^{(j)}$ ($j \not= j_0$)
and $\tilde{H}_k^{(j_0)}$ denotes the loop with vertex $r+j_0$. 
Thus, we observe that 
\[ \min \{  q(\tilde{G}) \mid \tilde{G} \in \mathcal{K}(\{H_k^{(j)}\}_{j \not= j_0}) \} 
	= \min \{ q(r), q_0(r) \}, \]
where 
\[ q_0(r) =  \min \{  q(\tilde{G}) \mid \tilde{G} \in \mathcal{K}(\{ \tilde{H}_k^{(j)}\}_{j=1}^{s(r)}) \}. \]
Therefore, \eqref{24/03/28/15:03} will be proven if the following holds:
\begin{equation}
\label{24/03/29/0:46}
q_0(r) \geq q(r). 
\end{equation}

We will only prove \eqref{24/03/29/0:46} for the case where $1 < j_0 < s(r)$;
\eqref{24/03/29/0:46} can be proven in a similar manner for $i_0=1, s(r)$.
The following cases are possible for $\tilde{G} \in \mathcal{K}(\{\tilde{H}_k^{(j)}\}_{j =1}^{s(r)})$:
\begin{itemize}
\item[(i)] There exists an $i_1$ such that $i_{j_0}- 1 \in U_k^{(i_1)}$ and $i_{j_0}+1 \in U_k^{(i_1+1)}$;
\item[(ii)] There exists an $i_1$ such that $i_{j_0} - 1 \in H_k^{(j_0-1)}$ and $i_{j_0} + 1 \in H_k^{(j_0+1)}$;
\item[(iii)] There exists an  $i_1$ such that $i_{j_0} - 1 \in U_k^{(i_1)}$ and $i_{j_0} + 1 \in H_k^{(j_0+1)}$;
\item[(iv)] There exists an $i_1$ such that $i_{j_0} - 1 \in H_k^{(j_0-1)}$ and $i_{j_0} + 1  \in U_k^{(i_1)}$.
\end{itemize}
We consider case (i) first. 
Since there is no cycle of length grater than $2$, 
we need only consider the following three subcases:
$\tilde{G} \in \mathcal{K}(\{\tilde{H}_k^{(j)}\}_{j = 1}^{s(r)})$
includes 
(a) the loop with the vertex $i_{j_0} \in V(\tilde{G})$,
(b) the cycle with the vertices $i_{j_0}$ and $i_{j_0}-1$, and
(c) the cycle with the vertices $i_{j_0}$ and $i_{j_0}+1$. 
If (a) holds, then so does the following:
\begin{equation}
\label{24.07.09.9:33}
q(\tilde{G}) \geq 2 + q(r).
\end{equation} 
Next, we consider the cases (b) and (c).
We define
\[ \tilde{U}_k^{(i)} 
	: = \begin{cases}
		U_k^{(i)}, & 1 \leq i \leq i_1 -1, \\
		U_k^{i_1} \cup \{ i_{j_0} \} \cup U_k^{i_1 +1}, & i=i_1, \\
		U_k^{(i+1)}, & i \geq i_1+1.
		\end{cases} \]
Then we have
\[ U_k \cup \{ i_{j_0} \} = \bigcup_{i=1}^{p(r)-1} \tilde{U}_k^{(i)}.  \]
Let 
\begin{equation*}
\tilde{q} = \#  \{ \tilde{U}_k^{(i)} \mid \mbox{$\#  \tilde{U}_k^{(i)}$ is odd} \}. 
\end{equation*}
By an argument  similar to the above, we can show that
\begin{equation}
\label{24/03/29/0:56} 
q(\tilde{G}) \geq \tilde{q}+1.
\end{equation}
Let us calculate $\tilde{q}$. 
Since $\#  \tilde{U}_k^{(i_1)}$ is odd if and only if 
both $\#  U_k^{(i_1)}$ and $\#  U_k^{(i_1+1)}$ are even or  both are odd,
we have
\[ \tilde{q}
	= \begin{cases}
		q(r)-1, & \mbox{if both $\#  U_k^{(i_1)}$ and $\#  U_k^{(i_1+1)}$ are odd,} \\
		q(r)+1, & \mbox{if both $\#  U_k^{(i_1)}$ and $\#  U_k^{(i_1+1)}$ are even,} \\
		q(r),   & \mbox{otherwise.}
		\end{cases}
\]
From \eqref{24/03/29/0:56}, we know that if (b) or (c) holds, then
\begin{equation}
\label{24.07.09.9:38} 
q(\tilde{G}) \geq q(r). 
\end{equation}
From \eqref{24.07.09.9:33} and \eqref{24.07.09.9:38},
we obtain \eqref{24/03/29/0:46} for case (i).

For case (ii), we observe that
$\tilde{G} \in \mathcal{K}(\{\tilde{H}_k^{(j)}\}_{j = 1}^{s(r)})$
includes the loop with the vertex $i_{j_0}$.
By an argument similar to the above, we have
\[ q(r) + 2 
	= q_0(r). \]
This establishes \eqref{24/03/29/0:46} for case (ii).

By an argument similar to the above, we know that in case (iii) 
$\tilde{G} \in \mathcal{K}(\{\tilde{H}_k^{(j)}\}_{j = 1}^{s(r)})$
must include include the cycle with vertices $i_{j_0}$ and $i_{j_0}-1$.
We set
\[ \hat{U}_k^{(i)} 
	: = \begin{cases}
		U_k^{(i)}, & i \not=i_1, \\
		U_k^{(i_1)} \cup \{ i_{j_0} \}, & i=i_1
		\end{cases} \]
and $\hat{q} = \#  \{ \hat{U}^{(i)}_k \mid \mbox{ $\#   \hat{U}^{(i)}_k$ is odd } \}$.
Then $q(\tilde{G}) \geq \hat{q} + 1$ holds. 
Since $\hat{U}_k^{(i)} $ is odd if and only if $U_k^{(i_1)}$ is even,
we have
\[ \hat{q} = \begin{cases} 
				q(r) + 1, & \mbox{if $U_k^{(i_1)}$ is even,} \\
				q(r)-1, & \mbox{if $U_k^{(i_1)}$ is odd.} 
		\end{cases} \]
This establishes \eqref{24/03/29/0:46} for case (iii),
and \eqref{24/03/29/0:46} is proven for case (iv) holds in a similar manner.
\end{proof}

\section{Spectrum}
\label{24/04/14/22:54}

In this section, we prove  parts (1) and (2) of Theorem \ref{24/02/29/17:13},
along with Theorem \ref{24/03/01/16:03}.

\subsection{Symmetry of the spectrum}
We first prove part (1) of Theorem \ref{24/02/29/17:13}.
Let $\lambda \in \sigma(L_G) \setminus \{0\}$.
It suffices to show that $-\lambda \in \sigma(L_G)$. 
From Proposition \ref{24/03/19.17:14}, 
$\lambda \in \sigma(L_G)$ if and only if there exists a $k \in [-\pi,\pi]$
such that \eqref{24/03/19.17:14} holds. 
Hence, we have
\[ \sum_{m=0}^{ (r+s(r)-\star)/2} A_m \lambda^{2m+\star} = 2 (\cos k) \lambda^{s(r)} \]
for some $k \in [-\pi,\pi]$. Then
\begin{align}
\sum_{m=0}^{ (r+s(r)-\star)/2} A_m (-\lambda)^{2m+\star} 
	& = (-1)^\star \sum_{m=0}^{ (r+s(r)-\star)/2} A_m \lambda^{2m+\star} \notag \\
\label{24/03/22/21:14}
	& = 2 (-1)^{s(r) + \star}(\cos k) (- \lambda)^{s(r)}. 
\end{align}
If $s(r) + \star$ is even, then \eqref{24/03/22/21:14} means that $-\lambda \in \sigma(L)$.
If $s(r)+\star$ is odd, then 
\[ \sum_{m=0}^{ (r+s(r)-\star)/2} A_m (-\lambda)^{2m+\star} 
	= 2 \cos (k\pm \pi) (- \lambda)^{s(r)},  \]
which also implies that $-\lambda \in \sigma(L_G)$.

\subsection{Spectral gap}
\label{24/03/28/21:06}

Next, we prove part (2) of Theorem \ref{24/02/29/17:13}.
From part (1) of Theorem \ref{24/02/29/17:13}, 
we know that $\rho(L_G)$ is symmetric with respect to $0$.
It suffices to show that
\begin{equation} 
\label{24/03/29/11:23}
 (0,\epsilon) \subset \rho(L_G)
\end{equation} 
for some (possibly small) $\epsilon > 0$.
From Propositions \ref{24/03/19.17:14} and \ref{24/03/27/10:36}, 
we know that $\lambda \in \rho(L_G)$ 
if and only if
\begin{equation}
\label{24/03/29/09:45} 
\sum_{m=m_0}^{ (r+s(r)-\star)/2} A_m \lambda^{2m+\star} \not= 2 (\cos k) \lambda^{s(r)} \quad
\mbox{for any $k \in [-\pi,\pi]$, }
\end{equation}
where $m_0 = (q(r)-\star)/2$.
Suppose that $\lambda > 0$. Then \eqref{24/03/29/09:45} is equivalent to
\begin{equation}
\label{24/03/29/09:57} 
\left( \sum_{m=m_0}^{ (r+s(r)-\star)/2} \frac{A_m}{2} \lambda^{2m+\star-s(r)} \right)^2 >  1.  
\end{equation}

We first suppose that $q(r)=0$. Then, from (a) of  Lemma \ref{24/03/29/09:59}, 
we know that
 $\star=0$ and the left hand side of \eqref{24/03/29/09:57} is
\[ \left( \sum_{m=0}^{ (r+s(r)-\star)/2} \frac{A_m}{2} \lambda^{2m-s(r)} \right)^2 
	= C_0 \lambda^{-2s(r)} + O(\lambda^{4-2s(r)}) \]
with $C_0 = (A_{m_0}/2)^2 = 1/4$.
Since $s(r) \geq 1$, 
this establishes \eqref{24/03/29/09:57} for sufficiently small $\lambda>0$.

Let us next assume that $q(r) \geq 1$. 
Then, from part (b) of  Lemma \ref{24/03/29/09:59},
it follows that
the left hand side of \eqref{24/03/29/09:57} is
\begin{equation}
\label{24/03/29/10:44} 
\left(  \sum_{m=m_0 + 1}^{ (r+s(r)-\star)/2} \frac{A_m}{2} \lambda^{2m+\star-s(r)} \right)^2 
	= C_1 \lambda^{2q(r)-2s(r)} + O(\lambda^{4+2q(r)-2s(r)}) 
\end{equation}
with $C_1 = (A_{m_0}/2)^2 > 1$. 
Since, by definition, $q(r) \leq p(r)$ and $p(r) \leq s(r)$,
we know that $q(r) \leq s(r)$.
If $q(r) = s(r)$, then the left hand side of \eqref{24/03/29/10:44} is
equal to $C_1  + O(\lambda^{4})$. 
In this case, since $C_1>1$, \eqref{24/03/29/09:57} holds for sufficiently small $\lambda>0$.
If $q(r) < s(r)$, then, because $2q(r)-2s(r) \leq -2$ and $C_1>0$, 
\eqref{24/03/29/09:57} holds for sufficiently small $\lambda>0$.
Thus, \eqref{24/03/29/11:23} is proven.


\subsection{Existence of the zero eigenvalue} \label{ss.3.3}

In this subsection, we prove Theorem \ref{24/03/01/16:03} and Proposition \ref{24/03/05/23:40.1}.
From Propositions \ref{24/03/19.17:14} and \ref{24/03/27/10:36},
we know that $\lambda$ is an eigenvalue of $L_k$ if and only if
\begin{equation}
\label{24/03/29/14:13} 
\lambda^{q(r)} F(k,\lambda) = 0,
\end{equation}
where $m_0 = (q(r)-\star)/2$ and
\begin{align} 
F(k,\lambda) 
& = \sum_{m=m_0}^{ (r+s(r)-\star)/2} A_m \lambda^{2(m-m_0)} - 2 (\cos k) \lambda^{s(r)-q(r)}
	\notag \\
\label{24/03/29/16:11}
& =  \sum_{m=0}^{ (r+s(r)-q(r))/2} A_{m+m_0} \lambda^{2m} - 2 (\cos k) \lambda^{s(r)-q(r)}.
\end{align}
From Lemma \ref{24/03/29/09:59}, we have
\[ F(k,0) = A_{m_0} - 2 \delta_{s(r),q(r)} \cos k \not=0, \
	\quad k \in [-\pi,\pi], \] 
where $$\delta_{s(r),q(r)} = \begin{cases} 1, & \mbox{if $s(r)=q(r)$,} \\ 0, & \mbox{if $s(r) > q(r)$.} \end{cases}$$
From \cite[Theorem XIII.85]{RS4}, 
we know that $L_G$ has $0$ as an eigenvalue if and only if
the Lebesgue measure of the set 
\[ \{ k \in [-\pi,\pi] \mid \mbox{$L_k$ has $0$ as an eigenvalue} \} \]
is positive. 
Hence, we observe that $0 \in \sigma_{\rm p}(L_G)$ is equivalent to $q(r) \geq 1$.
Since $q(r) \geq 1$ if and only if $\#  U_i$ is odd for some $i\in \{1,2,\cdots,p(r)\}$,
Theorem \ref{24/03/01/16:03} is proven.

To prove \eqref{24/03/29/15:15} in Proposition \ref{24/03/05/23:40.1}, 
we review some basic facts.
Let $L$ be a self-adjoint operator on a Hilbert space $\mathcal{L}$
and $\lambda_0$ an eigenvalue of $L$.
Then the geometric multiplicity $m_{\rm g}(\lambda_0) = {\rm dim}\ker (L-\lambda_0)$ of $\lambda_0$ 
is equal to the algebraic multiplicity  $m_{\rm a}(\lambda_0)$.
If $\mathcal{L}$ is a finite dimensional space
and $\lambda_i$ are eigenvalues of $L$, then
\[ {\rm det}(L - \lambda) = \prod_i(\lambda_i - \lambda)^{m_{\rm a}(\lambda_i)} \]
(see \cite[p.42]{Kato}).
From this general fact, we know that \eqref{24/03/29/15:15} holds.

\section{Band structure of the spectrum}

\subsection{Band functions}

We consider the solution $\lambda=\lambda(k)$ of 
\begin{equation}
\label{24/03/29/16:04} 
F(k, \lambda) = 0,
\end{equation}
where $F(k,\lambda)$ is a polynomial of degree $r+s(r)-q(r)$ as defined in \eqref{24/03/29/16:11}.
Let $\lambda_i(k)$ ($i=1,\cdots,r+s(r)-q(r)$) be solutions of \eqref{24/03/29/16:04}
and assume that
\[ \lambda_1(k) \geq \lambda_2(k) \geq \cdots \geq \lambda_{r+s(r)-q(r)}. \]
Then
\begin{equation}
\label{24/04/03/10:39} 
F(k,\lambda) = A_{(r+s(r)-q(r))/2} \prod_{i=1}^{r+s(r)-q(r)} (\lambda_i(k)-\lambda). 
\end{equation}

We also have the following lemmas:
\begin{lemma}
{\rm
$\lambda_1(0) = 1$ and 
\[ {\rm dim} \ker (L_0 - \lambda_1(0)) = 1. \]
}
\end{lemma}
\begin{proof}
From Lemma \ref{24.07/09/10:53}, $L_0$ is a stochastic matrix.
The well-known fact that the largest eigenvalue of a stochastic matrix is unique and equal to $1$
yields the desired result.  
\end{proof}
\begin{lemma}
\label{24/03/29/21:47}
{\rm
Let $\lambda_i(k)$ ($i=1,\cdots,r+s(r)-q(r)$) be as described above. 
Then the following hold for $i=1,2, \cdots, r+s(r)-q(r)$:
\begin{itemize}
\item[(1)] For any  $k \in (-\pi,0) \cup (0,\pi)$, ${\rm dim} \ker (L_k - \lambda_i(k)) =1$.
In particular, we have
\begin{equation} 
\label{24/034/03//10:19}
\lambda_1(k) > \lambda_2(k) > \cdots > \lambda_{r+s(r)-q(r)}(k),
	\quad k \in (-\pi,0) \cup (0,\pi). 
\end{equation}
\item[(2)] For $k=0, \pm \pi$, ${\rm dim} \ker (L_k - \lambda_i(k)) \leq 2$.
\end{itemize}
}
\end{lemma}
We postpone the proof of Lemma \ref{24/03/29/21:47} until the next subsection.

Because we know, from \eqref{24/03/29/16:11}, that $F(k,\lambda) = F(-k,\lambda)$, 
we have
\begin{equation} 
\label{24/04/02/10:32}
\lambda_i(-k) = \lambda_i(k), \quad k \in [-\pi,\pi]. 
\end{equation}
From \eqref{24/04/02/09:47} and \eqref{24/04/02/10:32}, we have
\begin{equation}
\label{24/04/03/01:39} 
\sigma(L_G) \setminus \{0\} 
	= \bigcup_{i = 1}^{r+s(r)-q(r)} \{ \lambda_i(k) \mid k \in [0,\pi] \}. 
\end{equation}
Therefore, we only require to study the behaviour of the solutions $\lambda_i(k)$ for $k \in [0,\pi]$.
\begin{lemma}
\label{24/04/03/16:01}
{\rm
Let $\lambda_i(k)$ ($i=1,\cdots,r+s(r)-q(r)$) be as described above.
Then $\lambda_i$ is analytic in a region that includes $(0,\pi)$ and is continuous in $[0,\pi]$. 
}
\end{lemma}
\begin{proof}
We first prove the first half of the lemma.
From \eqref{24/04/03/10:39} and \eqref{24/034/03//10:19}, we have
\[ \frac{\partial}{\partial \lambda}F(k,\lambda)|_{\lambda = \lambda_i(k)} 
	= -A_{(r+s(r)-q(r))/2} \prod_{j\not=i} (\lambda_j(k)-\lambda_i(k)) \not=0,
	\quad k \in (0,\pi) \]
for $i=1,\cdots,r+s(r)-q(r)$. 
Because $F(k,\lambda_i(k))=0$, 
we know from the implicit function theorem that 
$\lambda_i(k)$ is a unique analytic function in a region which includes $(0,\pi)$. 

We next prove the second half of the lemma. 
From Lemma \ref{24/03/29/21:47}, $\lambda_i(0)$ is a discrete eigenvalue of multiplicity at most $2$.
Because $\{ L_k \}_k$ is an analytic family  in the sense of Kato for $k$ near $0$
(see \cite[Theorem XII.13]{RS4}), we know that $\lambda_i(k)$ is continuous at $k=0$. 
The proof of the continuity at $k=\pi$ is similar.
\end{proof}
Now let 
\[ D(\lambda) 
	= \frac{1}{2} \sum_{m=0}^{t(r)} A_{m+m_0} \lambda^{2m-s(r)+q(r)}, \]
where $t(r)$ is defined in \eqref{24/04/02/23:55}.
We observe that $D(\lambda)$ is an even (resp. odd) function if $r$ is odd (resp. even),
because $s(r)-q(r)$ is even (resp. odd) if $r$ is odd (resp. even).
From \eqref{24/03/29/16:11} and \eqref{24/03/29/16:04}, 
we know that $\lambda = \lambda(k)$ is a solution of \eqref{24/03/29/16:04}
if and only if
\begin{equation} 
\label{24/04/002/10:36}
D(\lambda(k)) = \cos k.
\end{equation}
\begin{lemma}
{\rm
\label{24/04/03/01:43}
Let $\lambda_i(k)$ ($i = 1, \cdots, r+s(r)-q(r)$) be as described above.
Then:
\begin{itemize}
\item[(1)] If $r$ is even, 
\[ \lambda_{2t(r)+1-i}(k) = - \lambda_i(k),
\quad k \in [0,\pi], \quad i=1, \cdots,t(r). \]
\item[(2)] If $r$ is odd,
\[ \lambda_{2t(r)+1-i}(k) = - \lambda_i(k-\pi), 
\quad k \in [0,\pi], \quad i=1, \cdots,t(r). \]
\end{itemize}
}
\end{lemma}
\begin{proof}
We first prove (i). Suppose that $r$ is even.
Because, as mentioned above, $D(\lambda)$ is an even function in this case,
we have
\[ D(-\lambda_i(k)) = \cos k. \]
Hence, $-\lambda_i(k)$ is also a solution of \eqref{24/03/29/16:04}.
From Lemma \ref{24/03/29/21:47}, we obtain the desired result.

In the case where $r$ is odd, from an argument similar to that mentioned above,
we have
\[ D(-\lambda_i(k-\pi)) 
	= \cos k. \]
This proves (ii), again by Lemma \ref{24/03/29/21:47}.
\end{proof}

Next we present the

\begin{proof}[Proof of Theorem \ref{24/03/05/23:40}]
Because, from \eqref{24/04/002/10:36}, we have
\[ {\rm Arc} \cos D(\lambda_i(k)) = k, \quad k \in [0,\pi], \]
we know that 
$\lambda_i(k)$ ($ i=1,2,\cdots,r+s(r)-q(r)$) are strictly monotone functions in $k \in [0,\pi]$. 
From Lemma \ref{24/04/03/16:01}, we have
\[ \{ \lambda_i(k) \mid k \in [0,\pi] \} = [ a_i,b_i] \]
with $a_i = \min \{ \lambda_i(0). \lambda_i(\pi)\}$ and $b_i = \max \{ \lambda_i(0). \lambda_i(\pi)\}$.
Because we know from the first half of Theorem \ref{24/02/29/17:13} (2)
and  Lemma \ref{24/04/03/01:43} that
\[ \lambda_{t(r)}(k) \geq a_{t(r)} > 0 > b_{t(r)+1} \geq \lambda_{t(r)+1}(k), \]
it follows from \eqref{24/04/03/01:39} and Lemma \ref{24/04/03/01:43} once again that
\[ \sigma(L_G) \setminus\{0\} = \bigcup_{i=1}^{t(r)} [a_i, b_i] \cup [-b_i,-a_i]. \]
From \eqref{24/04/03/14:11}, we know that
\[ A_{m_0} = (-1)^{t(r)}(B_{m_0} + C_{m_0}).  \]
Let us first assume that $t(r)$ is even. 
Then, from Lemma \ref{24/03/29/09:59}, we have 
\[ \begin{cases} \lim_{\lambda \to 0} D(\lambda) = + \infty & \mbox{if $s(r) > q(r)$,} \\
	\lim_{\lambda \to 0} D(\lambda) >1 & \mbox{if $s(r) = q(r)$.}
	\end{cases} \]
From Proposition \ref{24/03/19.17:14}, we have $A_{t(r)+m_0} = A_{(r+s(r)-\star)/2} > 0$.
Hence, we obtain 
\begin{align*} 
& a_i 
	= \begin{cases}
	\lambda_i(\pi), & \mbox{if $i=1,3,\cdots, t(r)-1$,} \\
	\lambda_i(0), & \mbox{if $i=2,4,\cdots,t(r)$,} \end{cases} 
\\
& b_i 
	= \begin{cases}
	\lambda_i(0), & \mbox{if $i=1,3,\cdots, t(r)-1$,} \\
	\lambda_i(\pi), & \mbox{if $i=2,4,\cdots,t(r)$.} \end{cases}
\end{align*}
In the case where $t(r)$ is odd, we have
\[ \begin{cases} \lim_{\lambda \to 0} D(\lambda) = - \infty & \mbox{if $s(r) > q(r)$,} \\
	\lim_{\lambda \to 0} D(\lambda) <-1 & \mbox{if$s(r) = q(r)$.}
	\end{cases} \]
From an argument similar to that mentioned above, we know that
\begin{align*} 
& a_i 
	= \begin{cases}
	\lambda_i(\pi), & \mbox{if $i=1,3,\cdots, t(r)$,} \\
	\lambda_i(0), & \mbox{if $i=2,4,\cdots,t(r)-1$,} \end{cases} 
\\
& b_i 
	= \begin{cases}
	\lambda_i(0), & \mbox{if $i=1,3,\cdots, t(r)$,} \\
	\lambda_i(\pi), & \mbox{if $i=2,4,\cdots,t(r)-1$.} \end{cases}
\end{align*}
This  completes the proof of the theorem.
\end{proof}

\subsection{Absence of nonzero eigenvalues}
\label{24/04/04/16:38}

We now prove part (3) of Theorem \ref{24/02/29/17:13}.
From \cite[Theorem XIII.85]{RS4}, we may observe that 
$\lambda$ is an eigenvalue of $L_G$ if and only if the Lebesgue measure of the set 
$\{ k \in [-\pi,\pi] \mid \mbox{$\lambda$ is an eigenvalue of $L_k$} \}$
is positive. 
Hence, it suffices to show that
for any $\lambda \in \sigma(L_G) \setminus \{0\}$, the Lebesgue measure of 
\[ E_\lambda = \{ k \in [-k,k] \mid F(k,\lambda) = 0 \} \]
is 0.

Let us assume that $\lambda_0 \in \sigma(L_G)$. 
Then we know, from \eqref{24/04/03/01:39}, that there exists an $i$ 
such that $\lambda_0 \in \{ \lambda_i(k) \mid k \in [0,\pi] \}$.
Because $\lambda_i$ is strictly monotone in $[0,\pi]$, 
there exists a $k_0 \in [0,\pi]$ such that $\lambda_0  = \lambda_i(k_0)$
and $E_{\lambda_0} = \{ \pm k_0 \}$.
Hence, the Lebesgue measure of $E_{\lambda_0}$ is 0,
and part (3) of Theorem \ref{24/02/29/17:13} is proven.

\subsection{Multiplicity}

We still have to prove Lemma \ref{24/03/29/21:47}.
From \eqref{24/04/04/16:59}, we know that 
\[ {\rm dim}\ker(L_k - \lambda) = {\rm dim}\ker A_k(\lambda). \] 
By definition, we have
\begin{equation}
\label{24/04/04/21:06}
A_k(\lambda) 
= \begin{pmatrix}
	A_{11} & A_{12} \\ A_{21} & A_{22},
	\end{pmatrix}
\end{equation}
where $A_{11} = A_{11}(k,\lambda) \in M_r(\mathbb{C})$, 
$A_{21} =\!~ ^{\rm t}A_{12} \in M_{s,r}(\mathbb{C})$, $A_{22} \in M_s(\mathbb{C})$
and
\begin{align*}
& A_{11}(k,\lambda) 
= \begin{pmatrix} 
		-d_1 \lambda & 1 & 0 &  \cdots & 0 & e^{-ik} \\ 
		1 & -d_2 \lambda & \mbox{\rotatebox{140}{$\cdot\cdots$}} &  0 & \cdots & 0 \\
		0 & \mbox{\rotatebox{140}{$\cdot\cdots$}} & \mbox{\rotatebox{140}{$\cdot\cdots$}} & \mbox{\rotatebox{140}{$\cdot\cdots$}} 
			& \mbox{\rotatebox{140}{$\cdot\cdots$}} & \mbox{\rotatebox{90}{$\cdots$}} \\
		\mbox{\rotatebox{90}{$\cdots$}} & \mbox{\rotatebox{140}{$\cdot\cdots$}} & \mbox{\rotatebox{140}{$\cdot\cdots$}} 
			& \mbox{\rotatebox{140}{$\cdot\cdots$}} & \mbox{\rotatebox{140}{$\cdot\cdots$}}  & 0\\
		0 & \cdots & 0 & \mbox{\rotatebox{140}{$\cdot\cdots$}} & \mbox{\rotatebox{140}{$\cdot\cdots$}} & 1 \\
		e^{ik} & 0 & \cdots & 0 & 1 & -d_r \lambda
		\end{pmatrix}, 
\\
& A_{22}
	= \begin{pmatrix}
		-\lambda  & 0 &  \cdots & 0 \\ 
		0 & -\lambda & \mbox{\rotatebox{140}{$\cdot\cdots$}} & \mbox{\rotatebox{90}{$\cdots$}}  \\
		 \mbox{\rotatebox{90}{$\cdots$}}  & \mbox{\rotatebox{140}{$\cdot\cdots$}} & \mbox{\rotatebox{140}{$\cdot\cdots$}} & 0 \\
		0 & \cdots & 0  &  -\lambda \\
		\end{pmatrix}.
\end{align*}
Here $A_{12} = (e_{i,j})$ is given by
\[ e_{i,j} = \begin{cases} 1, & i=i_p, j = p \\ 0, & \mbox{otherwise} \end{cases}. \]
From direct calculation, we have
\[ A_{12}A_{21} = {\rm diag}(\delta_1, \cdots, \delta_r), \]
where
\[ \delta_i = \begin{cases} 1, & i = i_p, ~ p=1, \cdots, s, \\ 0, & \mbox{otherwise}. \end{cases} \]
Henceforth, we assume that $\lambda \not=0$. 
Then $A_{22}$ has the inverse $A_{22}^{-1} = \lambda^{-1} I \in M_s(\mathbb{C})$.
Hence, we obtain 
\begin{equation}
\label{24/04/04/21:14} 
A_{12} A_{22}^{-1} A_{21} = \lambda^{-1} {\rm diag}(\delta_1, \cdots, \delta_r). 
\end{equation} 
Let ${\bm x} = \!~ ^{\rm t} ({\bm x}_1, {\bm x}_2) \in \mathbb{C}^r \oplus \mathbb{C}^s$.
From \eqref{24/04/04/21:06} and \eqref{24/04/04/21:14}, 
we know that 
$A_k(\lambda){\bm x} = 0$ if and only if 
\begin{equation}
\label{24/04/04/21:23}
\begin{cases}
A_{11}{\bm x}_1 =  \lambda^{-1} {\rm diag}(\delta_1, \cdots, \delta_r){\bm x}_1, \\
{\bm x}_2 = -A_{22}^{-1} A_{21}{\bm x}_1.
\end{cases}
\end{equation}
From \eqref{24/04/04/21:23}, we observe that 
\[ {\rm dim}\ker(A_{11} - \lambda^{-1} {\rm diag} (\delta_1, \cdots, \delta_r))
	= {\rm dim}\ker A_k(\lambda). \]
To prove Lemma \ref{24/03/29/21:47}, 
it therefor suffices to examine the dimension of the kernel of the matrix
$E:=A_{11} - \lambda^{-1} {\rm diag} (\delta_1, \cdots, \delta_r)$.
Let $E {\bm x}_1 = 0$ and 
${\bm x}_1 = \!~ ^{\rm t}(x_1, \cdots, x_r) \in \mathbb{C}^r \setminus\{0\}$.
Then, from \eqref{24/04/04/21:23}, we have 
\begin{align}
\label{24/04/05/11:22}
& x_1 = e^{-ik} g_r x_r - e^{-ik} x_{r-1} \\
\label{24/04/04/22:35}
& x_2 = g_1 x_1 - e^{-ik} x_r  \\
\label{24/04/05/12:49}
& x_i = \lambda^{-1} (\delta_i + d_{i-1} \lambda^2)x_{i-1} 	- x_{i-2}, \quad 3 \leq i \leq r,
\end{align}
where we have set $g_i \equiv g_i(\lambda) := \lambda^{-1}(\delta_i + d_i \lambda^2)$, 
$i = 1, \cdots, r$.
Because, from \eqref{24/04/04/22:35} and \eqref{24/04/05/12:49}, 
we can write
\[ x_i = c_i^{(1)} x_1 + c_i^{(r)} x_r, \quad i = 1, \cdots, r \]
for some $c_i^{(1)}$ and $c_i^{(r)}$,
we have
\begin{equation} 
\label{24/04/05/14:03}
{\bm x}_1 
	= x_1 {\bm c}_1 + x_r {\bm c}_r  
\end{equation}  
with ${\bm c}_1 = \!~ ^{\rm t}(c_1^{(1)}, \cdots, c_r^{(1)})$ and 
${\bm c}_r = \!~ ^{\rm t}(c_1^{(r)}, \cdots, c_r^{(r)})$.
Hence, ${\rm dim}\ker E \leq 2$
and part (2) of Lemma \ref{24/03/29/21:47} is proven.

Next we prove part (1) of the lemma.
Let $x_i$ ($i=1,\cdots,r$) be as described above. 
From \eqref{24/04/05/12:49}, we know that
\begin{equation}
\label{24/04/05/11:28}
\begin{pmatrix}
x_r \\ x_{r-1}
\end{pmatrix}
	= K \begin{pmatrix}
			x_2 \\ x_1
			\end{pmatrix},
\end{equation}
where
\[  K = \begin{pmatrix}  
		g_{r-1}  & -1 \\ 1 & 0 
		\end{pmatrix}
			\cdots \begin{pmatrix}  
		g_2  & -1 \\ 1 & 0 
		\end{pmatrix}. \]
We also have
\begin{lemma}
\label{25.04.15.20:02}
{\rm
The following holds:
\begin{equation}
\label{24/04/05/13:15} 
( e^{ik} I - T ) \begin{pmatrix} x_1 \\ x_r \end{pmatrix} = 0, 
\end{equation}
where $T = (T_{i,j})$ is given by
\begin{align*} 
& T_{11} = g_1(g_r K_{11}-K_{21}) + g_r K_{12}-K_{22}, \\
& T_{12} = -e^{-ik}(g_r K_{11} - K_{21}), \\
& T_{21} = e^{ik}(g_1 K_{11} + K_{12}), \\
& T_{22} = - K_{11}
\end{align*}
with $K = (K_{i,j})$.
}
\end{lemma}
\begin{proof}
From \eqref{24/04/05/11:22}, 
we have
\begin{equation} 
\label{24/04/05/11:26}
	x_{r-1} = g_r x_r - e^{ik} x_1.
\end{equation}
Substituting \eqref{24/04/04/22:35} and \eqref{24/04/05/11:26} 
in \eqref{24/04/05/11:28},
we obtain
\begin{equation}
\label{24/04/05/13:40}
\begin{pmatrix}
0 & 1 \\ - e^{ik} &  g_r 
\end{pmatrix}
\begin{pmatrix}
x_1 \\ x_r 
\end{pmatrix}
 	= K 
\begin{pmatrix}
g_1  & - e^{-ik} \\ 1 & 0
\end{pmatrix} 
\begin{pmatrix}
x_1 \\ x_r
\end{pmatrix}. 
\end{equation}	
Several calculations establish \eqref{24/04/05/13:15}.	
\end{proof}			

By \eqref{24/04/05/11:26}, we have
\begin{align}
\label{24/04/05/13:52.1}
& (e^{ik} - T_{11}) x_1 - T_{12} x_r = 0, \\
\label{24/04/05/13:52.2}
& - T_{21} x_1 + (e^{ik}- T_{22}) x_r = 0.
\end{align}
By definition, $g_i$ ($i=1,\cdots,r$) are real numbers,
so are $K_{i,j}$, $T_{11}$ and $T_{22}$. 
Assume that $k \in (-\pi,0) \cup (0,\pi)$.
Then we have $(e^{ik} - T_{22}) \not=0$ and, by \eqref{24/04/05/13:52.2},
\begin{equation} 
\label{24/04/05/14:02}
x_r = \frac{T_{21}x_1}{ e^{ik} -T_{22}}. 
\end{equation}
If $T_{21}=0$, then $g_1 K_{11} + K_{12} = 0$.
Hence, we know from \eqref{24/04/05/13:40} that
\[ (1  + e^{-ik} K_{11} ) x_r = 0.  \]
Because $1  + e^{-ik} K_{11} \not=0$, we have $x_r = 0$.
From \eqref{24/04/05/13:52.1} and the fact that $e^{ik} - T_{11} \not=0$, 
we have $x_1=0$. 
However, this contradicts ${\bm x}_1 \not=0$.
Thus, we know that $T_{21} \not=0$. 
From \eqref{24/04/05/14:03} and \eqref{24/04/05/14:02},
we have ${\bm x}_1 = x_1 ({\bm c}_1 + (e^{ik}-T_{22})^{-1}T_{21} {\bm c}_r)$.
This completes the proof of Lemma \ref{25.04.15.20:02}.

\section{Two dimensional lattice}\label{twodim}

In this section, we first form a class $\mathscr{G}^d$ ($d \geq 1$) of graphs that are obtained from $\mathbb{Z}^d$
by adding pendant edges
in a manner similar to the way in which $G \in \mathscr{G}$ is obtained from $\mathbb{Z}$.
We will say that $G \in \mathscr{G}^d$ if $G$ satisfies the following:
\begin{itemize}
\item[(G$_1^{(d)}$)] $\{ ({\bm n},0) \in \mathbb{Z}^d \times \{ 0,1 \} 
		\mid {\bm n} \in \mathbb{Z}^d \} \subset V(G) \subset \mathbb{Z}^d \times \{ 0,1 \}$ 
		and a vertex $({\bm n},0) \in V(G)$ is connected to $({\bm n} \pm 1,0) \in V(G)$.
\item[(G$_2^{(d)}$)] 
	If $({\bm n},1) \in V(G)$, then $({\bm n},1)$ is an end vertex connected to $({\bm n},0)$ 
	and ${\rm deg}({\bm n},0)=2d + 1$.
	If $({\bm n},1) \not\in V(G)$, then ${\rm deg}({\bm n},0)=2d$.
\item[(G$_3^{(d)}$)] There exist vectors ${\bm a}_1, \dots, {\bm a}_d \in \mathbb{Z}^d$
linearly independent (as vectors in $\mathbb{R}^d$) such that
$({\bm n} + {\bm a}_1,1), \dots, ({\bm n} + {\bm a}_d,1)  \in V(G)$ 
	if and only if $({\bm n},1) \in V(G)$.
\end{itemize}
It is clear that $\mathscr{G}^1 = \mathscr{G}$ by definition.

Let $d=2$. 
We consider two graphs $G_j = (V(G_j), E(G_j)) \in \mathscr{G}^2$ ($j=1,2$)
satisfying the condition that
$({\bm n},1) \in V(G_j)$ if and only if
\[ {\bm n} = \sum_{i=1}^2 n_i {\bm a}_i^{(j)}, \quad (n_1, n_2) \in \mathbb{Z}^2, \]
where
${\bm a}_1^{(1)} = (2,0)$, ${\bm a}_2^{(1)} = (0,1)$,
${\bm a}_1^{(1)} = (2,0)$ and ${\bm a}_2^{(1)} = (1,1) \in \mathbb{Z}^2$.
Clearly, $G_j$ satisfies the condition (G$_3^{(d)}$) with ${\bm a}_i = {\bm a}_i^{(j)}$.
The following theorem reveals that the spectrum of graphs belonging to $\mathscr{G}^2$
is more complex than the spectrum of those belonging to $\mathscr{G}^1$.

\begin{theorem}
\label{05/11/16:46}
{\rm
Let $G_j \in \mathscr{G}^2$ ($j=1,2$) be
defined as above. Then:
\begin{itemize}
\item[(a)] $L_{G_1}$ has no spectral gap, 
	{\it i.e.} $\sigma(L_{G_1}) = [-1,1]$.
\item[(b)] $L_{G_2}$ has a spectral gap, 
	{\it i.e.} $\sigma(L_{G_2}) \not= [-1,1]$.
\end{itemize}
}
\end{theorem}
\begin{proof}
By an argument similar to that in the proof of of Proposition \ref{24/03/19.17:14},
we observe that $\lambda \in \sigma(L_{G_1})$ if and only if 
there exists  a pair $(k_1,k_2) \in [-\pi, \pi]^2$ such that 
\begin{equation}
\label{05/11/12:48}
20 \lambda^3 - 18 (\cos k_2) \lambda^2 -2(3- 2\cos^2 k_2) \lambda + 2 \cos k_2 - 2 (\cos k_1) \lambda = 0.
\end{equation}
Let $ u = \cos k_1$ and $v = \cos k_2$. Then we have
\begin{equation}
\label{06/19/16:17}
2 \lambda v^2 +  (1-9 \lambda^2) v +  \lambda (10 \lambda^2 - u - 3 ) = 0.
\end{equation}
We now consider a solution of the form $\lambda = v/a$. 
In this case, \eqref{06/19/16:17} reduces to
\begin{equation}
\label{06/20/8:54} 
v \{ (2a-5)(a-2) v^2 - a^2( u+3 -a) \} = 0 
\end{equation}
Let $a=2$. Then, from \eqref{06/20/8:54}  we observe that 
$\lambda = v/2$ is the solution of \eqref{06/19/16:17} 
for any $v$ in the range $-1 \leq v \leq 1$ if $u=-1$.
Hence we have $[-1/2, 1/2] \subset \sigma(L_{G_1})$.
Now let $a=1$. 
It is easy to observe from \eqref{06/20/8:54} 
that $\lambda = v$ is the solution of \eqref{06/19/16:17} whenever $|v| = \sqrt{(u+2)/3}$.
Since $-1 \leq u \leq 1$, we have $[-1, -1/\sqrt{3}] \cup [1/\sqrt{3},1] \subset \sigma(L_{G_1})$.
Finally, let $a = 1/2$.  
From \eqref{06/20/8:54} we can show that 
$\lambda = 2v$ is a solution of \eqref{06/20/8:54} if $|v| = \sqrt{(2u+5)/28}$.
Hence we have $[-\sqrt{7/12}, -1/2] \cup [1/2,\sqrt{7/12}] \subset \sigma(L_{G_1})$.
Thus, (a) is established.

Similarly to the proof of (a),
we observe that $\lambda \in \sigma(L_{G_2})$ if and only if 
there exists a pair $(k_1,k_2) \in [-\pi, \pi]^2$ such that 
\begin{equation}
\label{05/11/12:43}
\lambda \left\{ 20 \lambda^2 - 6-4 \cos k_2 - 4 \cos^2 k_2 - 2(1 + 2 \cos k_2) \cos k_1 \right\}= 0.
\end{equation}
We shall prove that $\lambda \in \rho(L_{G_1})$ if $0 < \lambda < 1/\sqrt{5}$.
We fix $\lambda_0 \in (0,1/\sqrt{5})$ and assume that $\lambda_0 \in \sigma(L_{G_1})$.
Let $u = \cos k_1$ and $v = \cos k_2$. 
Then we know, from \eqref{05/11/12:43}, that
there exist a pair $(u,v) \in [-1,1]^2$ such that
\[ v^2 + (1+u) v + \frac{3+u-10 \lambda_0^2}{2} = 0. \]
Because the above quadratic polynomial in $v$ has a solution, its discriminant satisfies
$u^2 + 20 \lambda_0^2 -5 \geq 0$.
On the other hand, by assumption, we have $u^2 + 20 \lambda_0^2 - 5 < 0$.
This is a contradiction, and (b) is therefore proven.
\end{proof}

{\bf Acknowledgments:} 
The author would like to thank Daiki Hayashi for useful discussions.
This work was supported by Grant-in-Aid for Young Scientists(Start-up) 22840022.


\begin{thebibliography}{99}
\bibitem{graphene} P. Chandrachud, B. S. Pujari, S. Haldar, B. Sanyal and D. G. Kanhere,
	A Systematic Study of Electronic Structure from Graphene to Graphane,
	\textit{J. Phys.: Condens. Matter} {\bf 22} (2010), 465502.
\bibitem{HiNo} Y. Higuchi and Y. Nomura, 
	Spectral structure of the Laplacian on a covering graph,
	\textit{European J. Combin.} {\bf 30} (2009), 570 -- 585.
\bibitem{Kato} T. Kato, 
\textit{Perturbation Theory for Linear Operators}, 2nd Edition,
		 Springer, Berlin Heidelberg New York (1976).
\bibitem{TS} A. Suzuki,
Spectrum of the Laplacian on a covering graph with pendant edges II:
	lattices in two dimensions,
	in preparation.
\bibitem{RS4} M. Reed and B. Simon,
\textit{Methods of Modern Mathematical Physics Vol. \Ro{4}},
        Academic Press, New York, 1978.   
  
\end{thebibliography}
\end{document}